\pdfoutput=1
\documentclass[final]{siamart190516}

\usepackage{ifthen}

\newboolean{useSISC}
\setboolean{useSISC}{true}

\ifthenelse{ \boolean{useSISC} }{
 \usepackage{hyperref}
}{
 \usepackage[T1]{fontenc}
 \usepackage{geometry}
 \geometry{left=4.1cm,textwidth=14.2cm,top=2cm,textheight=24cm}
 \usepackage{amsthm}
 \usepackage{amsfonts}
 \usepackage{array}
}

 \usepackage{graphicx}
 \usepackage{amssymb,amsmath}
 \usepackage{subcaption}

\usepackage[sort]{cite}

\usepackage{algorithm}
\usepackage{algpseudocode}
\usepackage{listings}

\newtheorem{observation}{Observation}

\newcommand{\particle}{p}

\newcommand{\timeStepSize}{\Delta t}

\newcommand{\particleSet}{\mathbb{P}}
\newcommand{\triangleSet}{\mathbb{T}}
\newcommand{\searchSet}{\mathbb{S}}

\newcommand{\halo}{\epsilon}
\newcommand{\particlePosition}{x}
\newcommand{\particleVelocity}{v}
\newcommand{\particleRotation}{r}
\newcommand{\particleAngularVelocity}{w}

\newcommand{\snapshotOld}{ t^{\text{(old)}} }
\newcommand{\snapshotCurrent}{ t^{\text{(current)}} }
\newcommand{\snapshotNew}{ t^{\text{(new)}} }
\newcommand{\snapshotCollision}{ t^{\text{(collision)}} }

\newcommand{\timeTerminal}{ t^{\text{(final)}} }

\newcommand{\timeMinCurrent}{ t^{\text{(current)}}_{\text{min}} }
\newcommand{\timeMaxCurrent}{ t^{\text{(current)}}_{\text{max}} }

\newcommand{\tria}{ \tau }

\begin{document}

\newcommand{\TheTitle}{Parallel local time stepping for rigid bodies represented by triangulated meshes}

\ifthenelse{ \boolean{useSISC} }{
 \title{
  \TheTitle
  \thanks{Submitted to the editors DATE.
   \funding{
The work was funded by an EPSRC DTA PhD scholarship (award no. 1764342).
It made use of the facilities of the Hamilton HPC Service of Durham University.
The research aligns and has been supported by EPSRC's Excalibur
programme through its cross-cutting project EX20-9 \textit{Exposing Parallelism: Task Parallelism}
(Grant ESA 10 CDEL) and the DDWG project \textit{PAX--HPC} (Gant EP/W026775/1).
Tobias' group appreciates the support by Intel's Academic Centre of
Excellence at Durham University.
   } 
  }
 }
 \author{
   Peter J.~Noble 
    \thanks{
     Department of Computer Science, 
     Durham University
     (\email{peter.j.noble@durham.ac.uk}).
    }
   \and 
   Tobias Weinzierl
    \thanks{
     Department of Computer Science, 
     Durham University
     (\email{tobias.weinzierl@durham.ac.uk}).
    }
 }
 \headers{Parallel Optimistic Local Time Stepping}{Peter J.~Noble and Tobias Weinzierl}
}{
  \title{\TheTitle}

  \titlerunning{Parallel optimistic local time stepping}

  \author{
   Peter J.~Noble 
    \thanks{
     Department of Computer Science, 
     Durham University
     (\email{peter.j.noble@durham.ac.uk}).
    }
   \and 
   Tobias Weinzierl
    \thanks{
     Department of Computer Science, 
     Institute for Data Science,
     Durham University
     (\email{tobias.weinzierl@durham.ac.uk}).
    }
  }
  \institute{Durham University}
}

\maketitle

\begin{abstract}
    Discrete Element Methods (DEM), i.e.~the simulation of many rigid particles,
suffer from very stiff differential equations plus multiscale challenges in
space and time.
The particles move smoothly through space until they interact
almost instantaneously due to collisions. 
Dense particle packings hence require tiny time step
sizes, while free particles can advance with large time steps.
Admissible time step sizes can span multiple orders
of magnitudes.
We propose an adaptive local time stepping algorithm which identifies
clusters of particles that can be updated independently, advances them
optimistically and independently in time, determines collision time stamps in
space-time such that we maximise the time step sizes used, and resolves the
momentum exchange implicitly.
It is combined with various acceleration techniques which exploit multiscale
geometry representations and multiscale behaviour in time. 
The collision time stamp detection in space-time in combination with the
implicit solve of the actual collision equations avoids that particles get
locked into tiny time step sizes, the clustering yields a high concurrency
level, and the acceleration techniques plus local time stepping avoid
unnecessary computations.
This brings a scaling, adaptive time stepping for DEM for real-world
challenges into reach.

\end{abstract}

\ifthenelse{ \boolean{useSISC} }{
 \begin{keywords}
  Discrete Element Method,
  Local time stepping,
  Multiscale contact detection,
  Task-based parallelisation
 \end{keywords}
 
 \begin{AMS}
  70E55, 70F35, 68U05, 51P05, 37N15
 \end{AMS}
}{
  \keywords{Discrete Element Method \and Local time stepping \and Multiscale
  contact detection \and Task-based parallelisation}
}

\section{Introduction}

%
%
In rigid body simulations using the Discrete Element Method (DEM) \cite{originalDEM},
particle-particle collisions describe the evolution of the system.
Collision detection is expensive and
usually forms the simulation's hot spot
\cite{multires2022,Iglberger:2010:GranularFlowNonSpherical,Li:1998:HierarchicalBoundingVolumes,Rakotonirina:2018:ConvexShapes}.
Therefore, expensive geometric checks should be performed if and only if
necessary,  
and they should run in parallel.
In an ideal world, only those checks identifying a collision are
performed, and we pick time step sizes that match exactly the time in-between two collisions, while particles are updated 
in parallel with minimal synchronisation.

%
%
Collision detection is particularly hard once the topology of the
particle interactions changes over time---if particles
for example settle into compact particle clusters and then spread out again.
Further to that, the shape of particles plays a crucial role
determining the behaviour of physical systems and the complexity of the
calculations.
Since simple analytical shapes
are insufficient to reproduce
large-scale phenomena of scientific
and engineering interest \cite{KRUGGELEMDEN2008153,ShapeSelection}, many codes
use compositions of analytical shapes.
Yet, they refrain from representing objects with 
(triangulated) meshes \cite{multires2022,Krestenitis:17:FastDEM} 
to remain computationally feasible.
Finally, supporting large variations in the sizes of particles 
is essential to understand many challenges of practical
relevance \cite{WEINHART2020107129}.

\begin{figure}[htb]
 \begin{center}
  \includegraphics[width=1.0\textwidth]{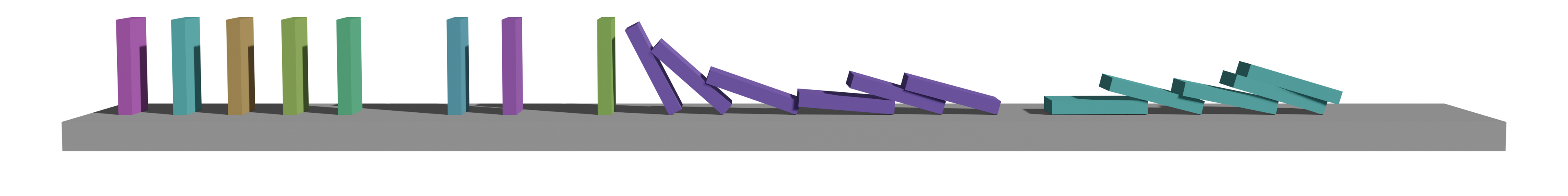}
 \end{center}
 \caption{
 Dominoes where multiple pieces (particles) tumble and push each
 other over. As long as particles do not move, they can advance in time
 with large time steps, while the collisions happen on short time scales
 and change the particle topology (who interacts with whom). Particle
 clusters which can use different time step sizes and advance in time
	 independently are shown in different colours.
  \label{figure:broad-phase-collision:clusters-with-static}
 }
\end{figure}

%
%
Complex, changing particle arrangements culminate in 
a globally stiff system of differential equations.
There is always two particles colliding, while the majority of
particles in the system are reasonably far away from each other and could
progress with larger time steps. 
It becomes impossible to find one reasonable time step
that suits all objects.
Furthermore, contacts affect only small regions
of the involved (large) particles, i.e.~a ``lot of geometry'' is not involved in
collisions at all.
Which and where particles interact changes permanently.
Continuous, implicit collision detection
\cite{Ferguson:2021:RigidIPC,CCDBench,Li2021CIPC,AffineDynamics}, where we
identify contact time stamps (almost) precisely, avoids the repeated
application of tiny time step sizes.
However, it is not trivial to generalise it to complex particle assemblies.
The combination of non-analytical rigid body representations for particles of
massively differing size, local time steps and implicit time
stepping ingredients is barely found in large-scale DEM
simulations.
They are considered to be too costly.

%
%
We propose a novel time stepping workflow, i.e.~a sequence of algorithmic steps,
to narrow the gap between what is computationally feasible and algorithmically desirable.
Hereby, each and every particle has its own time stamp and can advance 
with its own time step size.
First, we work with rough approximations of the particles which help us to
identify particle-particle pairs which can not yield collisions over a certain
time span.
We prune the graph of potential particle-particle interactions.
Second, we group particles into clusters which march in time in-sync. 
We give up on the individualism of the particles and group them, although
the grouping into clusters can change from step to step.
Clusters are chosen such that they cannot interact with each other, and
therefore can be handled in parallel.
Third, we determine an admissible time step size per cluster that is
just big enough to yield a first contact. 
We deal with the stiffness of the underlying differential equations with an
implicit scheme.
Finally, we allow the clusters of particles to progress, but if and only if it
remains possible for us to roll them back in time should we have missed out on
some collisions.

%
%
In three ways, we rephrase a computationally hard problem with an even more
challenging algorithmic language:
We rewrite the geometric problem to find contacts within a cluster in a
multi-resolution/-representation language, we abandon the idea of a
invariant clusters of particles, and we
switch to a space-time formulation.
These ideas can coexist since we carefully switch between implicit-explicit
formulations and conservative and optimistic time stepping, i.e.~techniques
that can guarantee that no physical constraints (incompressibility) are violated
and approaches that occasionally violate them but can, in such cases, roll back
and rerun.
The combination of implicit techniques with local time stepping avoids many tiny
global time steps \cite{adaptiveLocal18}, while we do not have to resort to any time step size discretisation \cite{localSPH}.
At the same time, the optimistic update of clusters keeps
the concurrency level high---few dense particle packings do not
throttle the overall simulation because they are taking tiny steps to catch up
\cite{localSPH,localTimeStepDEM}---while a synchronisation mechanism between
clusters ensures that we do not suffer from too many false negative contact
predictions \cite{CCDBench}.
In our opinion, the combination of these three elaborate ways to phrase the
computational challenge allows us to construct a significantly more efficient
and elegant algorithm than most state-of-the-art DEM time stepping schemes.

%
%
While our techniques bring simulations with large particle sets
represented by triangles into reach, the algorithmic mindset might
have potential impact beyond the regime of DEM, such as fluid-structure
interaction or other Lagrangian schemes.
It however yields a hard dynamic load balancing challenge.
On a single compute node with shared memory, 
work stealing can cope with it. The distributed memory world however
remains beyond scope here.

%
%
The remainder is organised as follows:
We start with a high-level overview of our algorithm, which allows us to
introduce our core physical formulae, algorithmic ingredients as well as the
used notation (Section~\ref{section:nutshell}).
This ``in a nutshell'' part is follow by the core methodological contribution of
the paper, where we first discuss a broad phase collision detection which
identifies the particle clusters (Section~\ref{section:cluster-identification}), 
the configuration and setup of the clusters
(Section~\ref{section:cluster-consolidation}), as well as the calculation of
their admissible time step size
(Section~\ref{section:time-step-size-calculation}).
To keep the simulation data consistent in time, we decide which particle
clusters are allowed to advance in time, before we actually compute the particle
interactions and allow them to progress in time (Section~\ref{section:time-stepping}).
With the algorithmic details at hand, we use 
Section~\ref{section:correctness} to discuss the algorithm's efficacy and
correctness.
Our performance results study various characteristics of
the arising algorithm, before we close the discussion in Section \ref{section:conclusion}.

\section{The algorithm}
\label{section:nutshell}

Let $ \particleSet $ denote our set of particles. 
Each particle is described by its geometry, its dynamic properties such as
velocity and rotation, as well as further meta data. 
Each of these quantities is parameterised through the particle's time
stamp, as the particles move and collide with each other.

\subsection{Physical model}

The ruling physical model is simple:
Each particle $\particle \in \particleSet $ is rigid, yet augmented 
by a small halo layer of width $\halo (\particle)>0$.
The halo serves as cushion weakening the notion of
a contact:
Two particles collide if their $\halo $-layers overlap.

If we have two
snapshots of the particle position $\particlePosition (p,t)$, velocity
$\particleVelocity(p,t)$, rotation $\particleRotation (p,t)$ and the angular
velocity $\particleAngularVelocity(p,t)$ for time
stamps $t_1 < t_2$, we can reconstruct all four quantities at
any point in-between through linear interpolation for the position and velocity
or spherical linear interpolation for quaternion types. 
We neglect higher order effects affecting the trajectory, and we omit
long-range and global forces such as gravity.
Each particle is studied at three snapshots $\snapshotOld (p) \leq
\snapshotCurrent (p) \leq \snapshotNew (p)$, i.e.~our intention is to develop
$\snapshotNew (p)$ from the two previous snapshots of all particles in the
system.

The time span between two snapshots is not constrained at all.
Each particle may reside on different time stamps.
Therefore, we
abandon the classic notion of a time step and 
instead consider a time step to be an operation which
advances a nonempty subset of $\mathbb{P}$.
To simulate the system's behaviour, we have to invoke the time step calculations
iteratively, until 

\begin{equation}
\timeMinCurrent = \min _{p \in \particleSet } \snapshotCurrent(p)
\label{equation:minimum:global}
\end{equation}

\noindent
fulfills $\timeMinCurrent \geq \timeTerminal$.
$\timeTerminal$
is the prescribed final simulation time.


\subsection{Algorithm blueprint} 
Each individual time step consists of six substeps:

\paragraph{Clustering.}
In a first substep, we identify 
clusters, i.e.~sets $\mathbb{P}_1, \mathbb{P}_2 \subset \mathbb{P}$ with $\mathbb{P}_1 \cap
\mathbb{P}_2 = \emptyset$, such
that $p_1 \in \mathbb{P}_1, p_2 \in \mathbb{P}_2$ implies that $p_1$ and $p_2$
can not collide over a time span $\timeStepSize $ of interest.
Let $c: \mathbb{P} \times \mathbb{P} \mapsto \{\bot, \top \}$ encode
if two particles might collide ($\top$) or definitely will not collide ($\bot$).

To make a decision if there might be a collision, we assume particles to be in 
free flight with the velocity at $\particleVelocity (p,\snapshotCurrent)$.
Therefore, we can extrapolate to a hypothetical state at $\snapshotNew =
\snapshotCurrent + \timeStepSize $ and check for 
collisions between any two particles along their hypothetical trajectories.
This yields the initial $c$.
It is symmetric, i.e.~$c(p_1,p_2) = c(p_2,p_1)$, and it
prunes the collision graph:
In theory, we have to assume that each particle could collide with every other.
The graph of potential particle-particle interactions is a clique. 
However, we know that it will be sparse in reality, as we work with
incompressible objects and small $\timeStepSize $.
Most particles cannot collide.
$c$ formalises this and prunes the graph.

The pruned, undirected graph hosts disconnected subgraphs.
All particles within a subgraph form one particle cluster.
As particles from different clusters cannot interact with each other, we handle
the particle sets in parallel and independently from hereon.

\paragraph{Cluster consolidation.}
In a second substep, we ensure that all particles assigned to one cluster
advance at the same pace.
In general, we have no guarantee that two particles $p_1$ and $p_2$ reside on
the same time stamp once we have identified a cluster $\particleSet
_i$ with $p_1,p_2 \in \particleSet _i$.
Therefore, we consolidate the particles within each cluster:
We compute the
cluster equivalent to (\ref{equation:minimum:global}), i.e.

\[
\timeMinCurrent (\particleSet _i) = \min _{p \in \particleSet _i} \snapshotCurrent(p),
\]

\noindent
and subsequently \emph{roll back} all particles within $\particleSet _i$ to this time
stamp by interpolating between $\snapshotOld$ and $\snapshotCurrent$.

\begin{definition}
 After the consolidation, a
 \emph{cluster} is a set of particles which all
 reside at the same $\snapshotCurrent$ and will advance with
 the same time step size. 
 \label{definition:cluster}
\end{definition}

\paragraph{Time step size calculation.}
In a third substep, we compute an admissible time step size per cluster subject to an admissible time span
$\timeStepSize $, which is, so far, only a crude upper limit for a time step
size.
Let $p_1, p_2 \in \mathbb{P}_i$ be two particles within one cluster.
The function $\snapshotCollision: \mathbb{P}_i \times \mathbb{P}_i \mapsto
[\timeMinCurrent (\particleSet _i), \timeMinCurrent (\particleSet _i) + \timeStepSize(\particleSet _i)] $ describes when the first collision within these
two particles happens.
$\snapshotCollision(p_1,p_2)=(\timeMinCurrent  + \Delta
t)(\particleSet _i) $ means that the two particles 
definitely do not collide over the time interval of interest or are, relative
to each other, at rest.
They might be stacked upon each other, e.g.
Other $\snapshotCollision(p_1,p_2)$ values mean the particles 
collide ``properly''.
We take the minimum of the latter values to determine a cluster's time step
size, i.e.~the minimal proper time step size which does not leapfrog any other
collision within the cluster:

\begin{definition}
 The \emph{effective time step size} of a cluster is chosen such that no two
 particles within this cluster collide before the end of the spanned time
 interval unless they are already at rest relative to each other.
 \label{definition:effective-time-step-size}
\end{definition}

\paragraph{Cluster masking.}
In a fourth substep, we ensure that clusters do not run ahead too quickly. 
This ensures that we can always roll back in time, i.e.~consolidate
the particles, if required.
The global reduction 

\begin{equation}
\snapshotCollision = \min _{\particleSet _i \subset \particleSet}
\snapshotNew(\particleSet _i)
  \label{equation:cluster-masking:min-collision-time}
\end{equation}

\noindent
indicates the earliest collision globally in this time step sweep.
The $\snapshotNew(\particleSet _i)$s result directly from the 
effective time step size calculation. 
Multiple collisions might happen later in time and might be updated as well in
this ``time step", but clusters with $ \snapshotCurrent
(\particleSet _i) > \snapshotCollision $ are disqualified from advancing:

\begin{definition}
 An \emph{active cluster} is a cluster which is allowed to progresses in time
 with its effective time step size, i.e.~is not masked out.
 \label{definition:active-cluster}
\end{definition}

\paragraph{Collision and update.}
In the fifth substep, we finally make the particles within
a cluster interact.
Algorithms with an $\epsilon$-formalism often replace the real contact force
equation described by a Dirac distribution with a smooth function whose support
is truncated to the halo.
We work with sequential impulses \cite{catto2006fast} instead.
Once an overlap of the $\halo$-regions is identified, we assume the point
within the overlap that sits right in-between
the two particles to be the 
exact contact point:
We temporarily ``expand'' the object geometries such that they fill out the
$\halo$ void, and then tailor the
interaction function such that the total momentum is preserved and the
particles are not squeezed together any further.
With the momentum exchange in place, we can update all particles within the
cluster.

\paragraph{Snapshot roll-over.}
Finally, we take those clusters that have advanced in time and roll them over:
$\snapshotOld \gets \snapshotCurrent$
and
$\snapshotCurrent \gets \snapshotNew$. All active clusters now
have advanced in time.
We evaluate the termination (or plotting) criteria and return to substep one for
the next time step.

\subsection{Correctness and efficacy}
Writing correct local adaptive time stepping methods for systems where each
particle is allowed to move at its own pace is non-trivial.

\begin{theorem}
 All particle snapshots at $t \leq \timeMinCurrent$ are \emph{valid snapshots}.
 If we advance the particles from $\particleSet $ in time, we know that no valid snapshot will
 ever turn out to be a wrong, i.e.~overly optimistic, guess.
 We notably will never have to roll back before $\timeMinCurrent $.
 \label{theorem:simulation-valid}
\end{theorem}

\noindent
As a direct consequence of the termination criterion 
(\ref{equation:minimum:global}), the simulation will deliver the correct
solution as long as the termination snapshot is reached:

\begin{theorem}
If we start a simulation with a given time stamp and apply our
time step algorithm iteratively,  all particles' $\snapshotCurrent$ eventually
will pass any later time stamp $\timeTerminal$ of interest.
\label{theorem:simulation-terminates}
\end{theorem}

\section{Broad phase collision detection to identify clusters}
\label{section:cluster-identification}

The function $c: \particleSet \times \particleSet \mapsto \{\top, \bot\}$
serves as guard for follow-up checks and sorts the particles into clusters.
Clustering is potentially expensive and
runs with quadratic complexity $\mathcal{O}(|\mathbb{P}|^2)$.
Our work hence
uses a combination of two techniques to find $\bot $ entries, i.e.~to narrow
down the number of potential collision pairs.

Let each particle move through space and time without any interaction. 
Each particle then spans a
space-time tube:
Let $B(p,\snapshotCurrent)$ be a bounding shape around a particle
at time $\snapshotCurrent$. 
It covers all of the particle including its $\halo $-halo.
We extrapolate this state to construct a bounding geometry 
$B(p,\snapshotNew)$.
The two bounding shapes can now be connected in space-time, i.e.~each object's
projected trajectory is now covered by a space-time bounding geometry
$B^{\text{(space-time)}}(p)$, and we can intersect these space-time bounding
objects of any particle pair $p_1,p_2$ to determine their $c(p_1,p_2)$ entry.
These entries are overly pessimistic if $B(p,\snapshotCurrent)$ is a poor
approximation, and likely wrong after a few collisions of two particles.
However, we only need it as pre-check and will later only advance up to the
first collision of two particles anyway.

\paragraph{Particle-specific time step sizes.}

No two particles might reside on the same time stamp.
Therefore, it is important to make the space-time
formalism incorporate $\snapshotOld$, too:
We study a space-time bounding geometry from $t \in (\snapshotOld,
\snapshotCurrent) \cup (\snapshotCurrent, \snapshotNew)$ which consists of two
interpolated segments which are concatenated.
One bounding geometry might ``start earlier'' than the other particle's.
Also the span of the geometry in time can differ.
Let

\begin{equation}
 \timeStepSize(p_i)  = \min(
    \alpha (\snapshotCurrent(p_i) - \snapshotOld(p_i)),
    \timeStepSize
   ) 
 \label{equation:cluster-step:damped-particle-time-step-size}
\end{equation}

\noindent
identify a particle's time step size at this point.
Our code implicitly memorises the previously used time step through the two time
stamps $\snapshotCurrent(p_i)$ and $\snapshotOld(p_i)$.
These quantities are stored per particle.
Due to the constant $\alpha>1$,
the time step size that feeds into the
trajectory prediction becomes a creeping average between the maximum global time
step size $\timeStepSize $ and previous time steps.

Follow-up steps have to take into account that the $c$-statements
refer to a particle-specific time span only.
Notably, $c$ only makes statements over  

 \vspace{-0.25cm}
\begin{equation}
 \Big(
   \timeMinCurrent(p_1,p_2)
   ,
   \timeMinCurrent(p_1,p_2)
   + 
   \min(\timeStepSize(p_1), \timeStepSize(p_2))
  \Big), 
  \quad
  \text{where}
  \label{equation:nutshell:max-time-span-between-two-particles}
\end{equation}

\vspace{-0.15cm}
\[
    \timeMinCurrent(p_1,p_2) =
    \min(\snapshotCurrent(p_1),\snapshotCurrent(p_2)),
\]

\noindent
as we construct the bounding geometries over 
$\snapshotOld(p_i)$, $\snapshotCurrent(p_i)$ and $\snapshotNew(p_i)$, where the
latter is determined the particle-specific time step size.

\begin{figure}[htb]
 \begin{center}
  \includegraphics[width=0.5\textwidth]{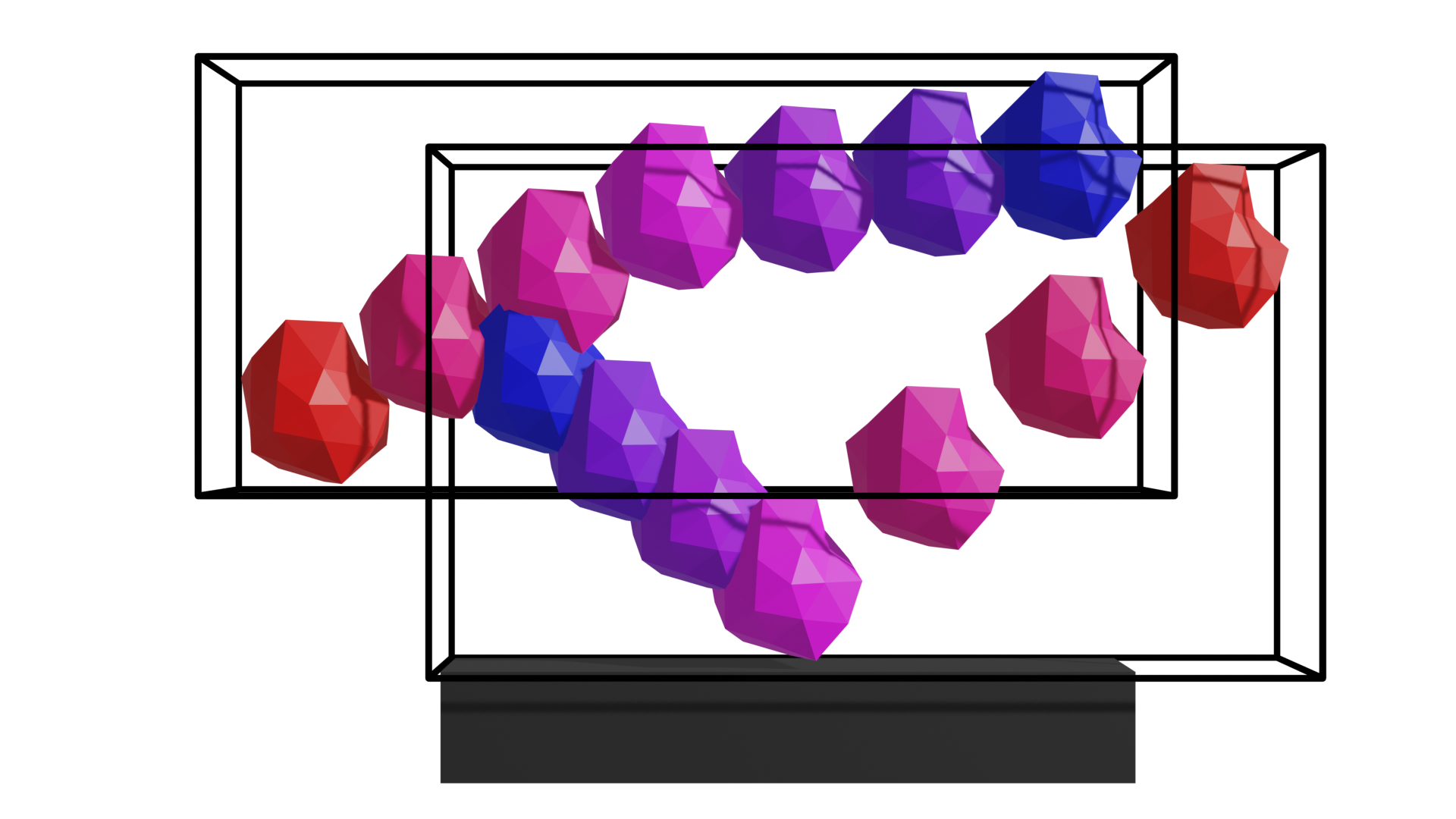}
 \end{center}
 \caption{
  The last state (red), current state (pink) and future state (blue) and the overall bounding box
  for two particles in motion.
  One particle moves left to right following an arc.
  The other particle moves right to left and bounces on the ground.
  The overall bounding boxes of both objects suggests that a collision might occur.
  However, when we consider the space-time volume spanned by the particles it is clear that these particles will pass by each other without interacting.
 }
\label{figure:broad-phase-collision:bounding-box}
\end{figure}

\paragraph{Check 1: Space-time bounding boxes.}

With a well-defined time step size per particle at hand, our first technique to
construct $c$ uses axis-aligned bounding boxes for $B(p)$ and constructs a
space-time bounding box for $B^{\text{(space-time)}}(p)$. These space-time bounding boxes between
any two particles are compared.

As we use a space-time bounding box,
the algorithm can be rewritten as a purely spatial approach.
The space-time check collapses along the time dimension (Figure
\ref{figure:broad-phase-collision:bounding-box}): 
We take the particle positions at $\snapshotOld(p_i)$ and
$\snapshotCurrent(p_i)$ and surround both
including their $\halo$-environment with a bounding box. 
If this super-bounding box of two particles does not overlap, $c(p_1,p_2) =
\bot$.
The handling of $\snapshotCurrent(p_i)$ and
$\snapshotNew(p_i)$ follows the same pattern.

\paragraph{Check 2: Space-time tubes.}

Having two space-time bounding boxes overlapping can result in a lot of
false-positives, i.e.~potential collisions which are not really confirmed later
on.
This happens, for example, if two fast moving particles move parallel to each
other.
The second particle might take the first particle's position after one time
step, while the first particle has moved on meanwhile.
To eliminate such false-positives, a second test 
doublechecks and revises $c(p_1,p_2)=\top$ entries.

This second technique embeds each particle including its $\halo $ into a
bounding sphere $B$.
The space-time object $B^{\text{(space-time)}}(p)$ now resembles a hose with two 
linear segments, as we
model each particle as a sphere moving linearly along 
$\snapshotOld$ to $\snapshotCurrent$ and then from $\snapshotCurrent$
to $\snapshotNew$. 
If the hoses do not
intersect, we reset $c(p_1,p_2)=\bot$.
Otherwise, this second test confirms $c(p_1,p_2)=\top$.

\paragraph{Static geometric objects.}
The treatment of static geometry parts requires special attention throughout the
clustering:
Floors for example are often modelled with few huge triangles
(Figure~\ref{figure:broad-phase-collision:clusters-with-static}).
Without special rules, they ``couple'' all particle assemblies hitting
this wall globally, even though the assemblies of particles could be handled as
separate clusters.
We therefore treat static objects specially:

\begin{definition}
 Static geometric objects are parts of all the clusters which host at
 least one particle which directly interacts with them.
 However, their collisions do not feed into the cluster identification.
\end{definition}

\paragraph{Efficiency, contextualisation and implementation.}
The implementation of the two-step check to compute $c$ has to be fast. Our 
realisation relies on Intel Embree's raytracing kernels \cite{Embree}.
This raytracing software is tuned to handle bounding
box checks and variants thereof efficiently.
For the actual clustering, i.e.~the identification of unconnected subgraphs,
a parallel iterative method is used:
Nodes are repeatedly coloured with the lowest colour id of their neighbours plus
one, until a further iteration produces no changes anymore.

The first pre-check with cubes stands in the tradition of pessimistic
algorithms, and yields information of limited interest if particles move very fast or over a
long time span, while the second step actually takes rapid changes of the
trajectory and fast velocities into account.
It prevents too many collision pairs to be added to $c$, while the first step
handles quasi-stationary setups quickly and eliminates them from further checks.

Stationary objects such as walls are explicitly omitted from the initial
clustering.
They are invisible to the extrapolation. 
Once the clustering algorithm has terminated, we replicate
all static objects, such that each cluster sees ``its own version'' of the static object such as a
wall.
Many particles that hit the same wall in some distance thus tend to end up in
different clusters whereas the wall as proper geometric particle in the
computation of $c$ would merge all of these objects into one big cluster and
remove concurrency.

The efficiency of our overall algorithm depends upon the fact if we
eventually manage to use as big time steps as possible to reduce the total
number of time steps taken.
The time span feeding into the clustering serves as natural constraint for the
final time step size.
A too large span hampers the expressiveness of $c$ unless particles are
in totally free flight.
We hence make is depend on historic data.
Particles which previously have used a very small time step size
will continue to try out rather small time steps, while others might use a time
step size close to $\timeStepSize$.
Through $\alpha $, we can control 
how quickly a previously slow particle
approaches the maximum time step $\timeStepSize$ and hence steer the
``optimism'' in the prediction.
We use $\alpha = 2$.

The global $\timeStepSize $ is in itself a tuning
parameter without physical meaning, which controls the optimism of our time
stepping.
By picking small $\timeStepSize $, we run risk to constrain those particles
overly pessimistically which could travel free of any collision over longer time
spans.
Large $\timeStepSize $ however make the $c$ checks overly pessimistic and lead
to large clusters.

All steps in this first phase yield classic data parallelism (parallel fors),
often combined with a Boolean reduction.
The individual checks (per particle pair, e.g.) are reasonably cheap such that
further inner parallelisation seems to be unreasonable.
However, once identified, clusters form independent units of work.
We can handle each cluster as
an independent task within this step and from hereon.
All tasks and data parallelism are realised via Intel's TBB \cite{TBB}.

\begin{observation}
 Particle sets that form a cluster can be treated independently from each
 other from hereon as separate tasks. 
\end{observation}

\section{Cluster setup and consolidation}
\label{section:cluster-consolidation}

As we want to advance all the particles within a cluster in-sync, we 
have to ensure that they all start from the same snapshot.
Our clusters are re-determined in each and every sweep.
Consequently, a common time stamp $\snapshotCurrent$ for all particles of a cluster
is not guaranteed once the cluster identification in
Section~\ref{section:cluster-identification} terminates.
We hence determine $\timeMinCurrent (\particleSet _i)$ from
(\ref{equation:minimum:global}) and interpolate between the particles'
snapshots to reconstruct a joint, synchronised state.
From hereon, $\timeMinCurrent (\particleSet _i) = \snapshotCurrent (p_i) \
\forall p_i \in \particleSet _i$.

Further to the consolidation of the particle snapshot, we also assign each
cluster $\particleSet _i$ a unique
\begin{equation}
 \timeStepSize (\particleSet _i) = \min _{p \in \particleSet _i} \timeStepSize
 (p).
 \label{equation:cluster-consolidation:cluster-time-step}
\end{equation}

\noindent
This is a quantity that is used in follow-up steps.
It is a restricted value over the cluster fed by
historic data of the individual particles.

\paragraph{Efficiency, contextualisation and implementation.}

Per cluster, we run through all particles, but each particle is handled
independently.
It yields a prime example for data parallelism within the cluster task, where we
reduce two scalar quantities, i.e.~time step size and time stamps.

%

\begin{observation}
 Whenever particles change cluster membership, they typically are \emph{rolled
 back} in time.
\end{observation}

\noindent
Our cluster consolidation substep corrects overly optimistic time step choices.
However, this is not a rollback in the sense that we return to a previous time
stamp.
Instead, we interpolate, i.e.~we only move slightly backwards in time.

As we reduce the time step size of individual particles, we potentially sparsify
$c$ further.
This can lead to situations where a cluster decomposes into several subclusters
which in turn increases the concurrency level.
We did not find any performance gain from a re-clustering and
further decomposition of cluster tasks, and therefore do not exploit this
additional increase of concurrency.

We apply a physical simplification as
we assume that we can interpolate.
Furthermore, the interpolation relies on the assumption that a particle
does not experience any collisions between $\snapshotOld$ and $\snapshotCurrent$
(cmp.~Theorem~\ref{theorem:simulation-valid}).

\section{Time step size calculation}
\label{section:time-step-size-calculation}

The time step calculation identifies the
largest time step size with which all particles of that  cluster can safely
advance without missing out on further collisions.
Let $\snapshotCollision(p_1,p_2) $ denote the time when two particles'
halos intersect for the first time. It is subject to

\vspace{-0.25cm}
\begin{eqnarray}
 \timeMinCurrent (\particleSet _i) \ll \snapshotCollision(p_1,p_2) 
   & \leq &
   \timeMaxCurrent
   \quad
   \forall p_1, p_2 \in \particleSet _i,
   \label{equation:time-step-size-calculation}
   \\
 \text{where} \qquad
   \timeMaxCurrent (\particleSet _i) & = & \timeMinCurrent (\particleSet _i) + 
   \timeStepSize (\particleSet _i).
   \nonumber
\end{eqnarray}

\noindent
The following case distinctions feed into the calculation of
$\snapshotCollision(p_1,p_2)$: 

\begin{enumerate}
  \item If two particles move away from each other,
  $\snapshotCollision(p_1,p_2) = \timeMaxCurrent (\particleSet _i)$. Particles
  that separate at $\snapshotCurrent(\particleSet _i)$ do not constrain the
  collision search.
  \item If two particles' $\halo$-area overlaps at
  $\snapshotCurrent(\particleSet _i)$, 
  $\snapshotCollision(p_1,p_2) = \timeMaxCurrent (\particleSet _i)$. Particles
  that stick together at the begin of a time stamp (they rest upon each other
  or move parallel, e.g.) do not constrain the time step size further.
  \item For all other particles, we compute the collision time stamp.
\end{enumerate}

\noindent
With (\ref{equation:time-step-size-calculation}), we can update
each cluster's time step or new time stamp
respectively:

\begin{eqnarray*}
 \snapshotNew (\mathbb{P}_i)
 & = &
 \min _{p_1,p_2 \in \mathbb{P}_i} \snapshotCollision(p_1,p_2)
 \ < \ 
 \Delta t,
 \qquad \text{and therefore} \\
 \Delta t (\particleSet _i) & \gets & \snapshotCollision
(\mathbb{P}_i) - \snapshotCurrent (\particleSet _i).
\end{eqnarray*}


\paragraph{Geometric model.}

Let $ \triangleSet _h(p) $ describe the set of triangles describing a particle
$p$.
$\tria ^{\halo}$ denotes the volumetric extension of its corresponding triangle
$ \tria \in \triangleSet _h(p,t) $, i.e.~the triangle plus its $\halo $-environment.

\[
 \triangleSet _h ^{\halo} (p,t)
 = 
 \bigcup _{\tria \in \triangleSet _h(p) } \tria ^{\epsilon}
\]

\noindent
yields a shell object, i.e.~a volumetric geometric object into which the surface
of the particle is embedded.


\begin{definition}
 A \emph{contact point} is a tuple of a position in space $x$, a time stamp $t$,
 a normal vector $n$ and two triangles $\tria _1 \in p_1$ and $\tria _1 \in p_2$
 from two different particles $p_1 \not= p_2$.
 The following properties hold:
 \begin{enumerate}
   \item The distance between $x$ and $\tria _1$ or $\tria _2$ is at most
   $\halo$.
   \item Two triangles have at most one
   contact point per time stamp.
   \item $n$ is the shortest distance vector to one of the nearest triangles.
 \end{enumerate}
\end{definition}

\noindent
Two particles are in contact at a certain time, if there is at least one contact
point between their triangles.
We note that two triangulated particles can yield redundant contact points if
any contact lies on a shared edge or vertex and thus arises from multiple
triangle pairs.
Contact points furthermore are not unique for parallel triangles---they
can be located anywhere on a submanifold within the $\halo$-overlap.
Our discussion from hereon assumes that contact points are filtered,
i.e.~contact points that are close in space are fused into one contact point to
avoid spatial replicas and are centred within parallel triangles.

\paragraph{Collision time detection.}

Contacts between two triangles either are placed on the shortest
distance between a vertex of one triangle and the other triangle's surface, or are located in-between two triangle edges.
Efficient code blocks to compute them are known
\cite{multires2022,Krestenitis:17:FastDEM}.
We have to transfer such spatial algorithms into
the space-time domain over $(\snapshotCurrent,
\timeMaxCurrent)(\particleSet _i)$ with a new (minimal)
$\timeStepSize (\particleSet _i)$ in-between to be determined.
For this, we employ an iterative algorithm.

\begin{definition}
 We compute an effective time step size
 (cmp.~Definition~\ref{definition:effective-time-step-size}) from $\timeStepSize
 (\particleSet _i)$ by reducing this time span until any further reduction of
 the time step size would yield no collisions within the cluster. This process
 is called \emph{narrowing}.
 \label{definition:narrowing}
\end{definition}

\noindent
We realise the narrowing iteratively by repeatedly reducing the time span in
which we search for collisions.
Per iteration per triangle pair, we treat vertex-triangle and edge-edge
comparisons separately:
For each vertex-triangle pair, we compute the position of the vertex
relative to the triangle at the begin and end of the maximum admissible time step:
We use a relative coordinate system which moves with the triangle.
After that, the algorithm takes the space-time line segment connecting
the two relative positions.
We can now analytically compute the minimum distance, as the minimum distance is
either observed at the start or the end point, or arises from
the time stamp when the line intersects the (relative) plane in which triangle
lives---if it exists.
For each of these two or three, respectively, situations we can use Barycentric
coordinates to construct the actual contact point plus its time stamp.
It is an approximation as we neglect rotation.
For the edge-edge comparisons, we lack an analytic formula and hence sample
the initial distance, the final distance and the distance in-between.
Both types of distances might identify a contact with a certain time stamp.
As long as the minimum of these contact time stamps remains smaller than the 
time span of interest, we narrow the span just to include this minimum and rerun
the tests.

\paragraph{Accelerated multiscale algorithm.}

To accelerate this narrowing, we use an iterative
multi-resolution scheme exploiting the notion of surrogate geometries
\cite{multires2022}:

\begin{definition}
 Let a \emph{surrogate} representation of a triangle set $\triangleSet _h
 ^{\halo} (p)$ be another triangle set $\triangleSet _{2h} ^{\halo} (p)$ with 
 \begin{eqnarray}
   \triangleSet _h ^{\halo} (p) & \subseteq & \triangleSet _{2h} ^{\halo} (p)
   \ \text{with} \ 
   \forall \tria _h \in \triangleSet _h (p): \ \exists ! 
   \tria _{2h} \in \triangleSet _{2h} (p)
   \ \text{s.th.} \ 
   \tria _{h}^\halo \subseteq \tria _{2h}^\halo,
   \label{equation:surrogate:conservative}
   \ \text{and} \\
   | \triangleSet _h ^{\halo} (p) | & > & | \triangleSet _{2h} ^{\halo}
   (p) |.
   \label{equation:surrogate:efficient}
 \end{eqnarray}
 
 \noindent
 The $\tria _{2h}$ in (\ref{equation:surrogate:conservative}) is the
 \emph{parent} of its corresponding $\tria _h$.
\end{definition}

\noindent
Per particle, we hold a sequence (levels) of surrogates
$\triangleSet _h ^{\halo} (p), \triangleSet _{2h} ^{\halo} (p)$, $\triangleSet _{4h} ^{\halo}
(p), \ldots$.
Each surrogate is conservative with respective to the next finer surrogate due
to (\ref{equation:surrogate:conservative}):
If two surrogates of level $k$ do not collide,
their corresponding surrogates of level $k/2$ do not collide either.
Each surrogate is also effective with respect to the next finer surrogate due to
(\ref{equation:surrogate:efficient}):
It is cheaper to handle a surrogate of level $k$ compared to the
surrogate of level $k/2$, as the former features fewer triangles. 

As (\ref{equation:surrogate:conservative}) formalises a surrogate tree, 
we can introduce an efficient multilevel iterative variant
expanding upon our vanilla code version:

\begin{enumerate}
  \item For two particles $p_1$, $p_2$ of interest, we start with their coarsest
  surrogate representations. They define the search sets $\searchSet (p_1)$ and
  $\searchSet (p_2)$.
  \label{algorithm:time-step-size-calculation:kick-off}
  \item Per pair $(\tria _1, \tria _2) \in
  \searchSet (p_1) \times \searchSet (p_2)$, we approximate the minimum time of
  contact. $\tria_1$ and $\tria_2$ are removed from their search sets once we
  have evaluated all the pairs they are involved in.
  \label{algorithm:time-step-size-calculation:iteration}
  \item If a new contact times exists and it falls into our search span defined
  by $\snapshotCurrent (\particleSet _i)$ and $\snapshotCollision (\particleSet
  _i)$,
  \label{algorithm:time-step-size-calculation:add-data}
  \begin{itemize}
    \item we update $\snapshotCollision (\particleSet _i)$ if they reduce this
    quantity further and if $\tria _1$ and $\tria _2$ are both taken from
    the fine mesh triangulations, i.e.~$(\tria _1, \tria _2) \in \triangleSet
    _{kh} ^{\halo} (p_1) \times \triangleSet _{kh} ^{\halo} (p_2)$;
    \item otherwise, we throw away these contacts and the contact times and
    instead insert all children of $\tria _1$ and $\tria _2$ into $\searchSet (p_1)$ or
    $\searchSet (p_2)$, respectively.
  \end{itemize}
  \item If $\searchSet (p_1) \not= \emptyset \wedge \searchSet (p_2) \not=
  \emptyset $, we continue with
  Step~\ref{algorithm:time-step-size-calculation:iteration}.
\end{enumerate}

\paragraph{Efficiency, contextualisation and implementation.}

Our particle-to-particle comparison only studies particle combinations with
$c(p_1,p_2) = \top$.
Even though two particles reside within one cluster, we might have come to the
conclusion earlier that these particles cannot collide.
This information is used here.

Within each particle-particle comparison, the multiscale variant terminates the
exploration of branches of the surrogate trees early.
Even triangle pairs that would result in a collision later within the search 
time span are omitted, once we know that another triangle pair yields a
collision earlier on.
We reduce the number of triangles to study within the tree, but also decrease
$\snapshotCollision (\particleSet _i)$ monotonously.
Therefore, it makes sense to study all first levels of all particles within a
cluster first, 
before we switch to the next finer resolution for further geometry checks.
For this later unfolding, we might already have reduced 
$\snapshotCollision (\particleSet _i)$ and hence study fewer triangles from the
second resolution level.

We assume that triangles move linearly through space.
This is not valid once $\particleRotation (p) \not=0 $.
Our current implementation ignores the possibility that a section of 
an object rotates fully through another one.
For small admissible time step sizes and slow rotations, this is reasonable.
In an application with finer interacting geometry (for example, fine teeth on interlocking gears)
it may be necessary to add extra checks here.
An alternative option is it to increase the $\halo$ region of each triangle and
surrogate to ensure the real volume swept by a rotation triangle is covered.
Notably, we can make $\halo $ the larger the coarser the surrogate.
While this would introduce more false-positive collisions early throughout the
multiscale algorithm, it would ensure there are no false-negatives.
More sophisticated geometric methods exist \cite{CCDBench} which might
be able to provide additional efficiency and correctness guarantees for such
challenging cases.

Experimental evidence suggests that floating point round-off errors lead to
situations, where the algorithm yields no contact points
at all, as we just slightly underestimate a valid $\timeStepSize (\particleSet
_i)$.
Further to that, we will need a certain non-zero overlap of the particle
$\halo$s in the subsequent step.
Therefore, we make our algorithm artificially increase the time step size
slightly after each iteration
This is admissible, as we work with weakly compressible objects subject to an
$\halo $-layer anyway.
Inspired by static code analysis in Definition~\ref{definition:narrowing}, this
slight increase of the time step size should be labelled \emph{widening}.

\begin{observation}
  Even though we search for a minimum contact time per cluster, multiple
  contacts can arise.
  \label{observation:time-step-size-calculation:contact-approximation}
\end{observation}

\noindent 
The algorithm's underlying sweeps over particle pairs and triangle pairs
translates directly into data parallelism. 
However, a realisation with parallel for loops over the surrogate levels is
disadvantageous as it introduces synchronisation points after each
level.
It also suffers from non-constant compute cost per compute step.
We hence advocate for a strict dynamic task tree formalism, processing
triangle-triangle pairs within the surrogate hierarchy independently.
The unfolding of this task tree is constrained by a global (atomic) contact
time, i.e.~we might unfold too many tree nodes.

\section{Particle interaction}
\label{section:time-stepping}


Our physical model approximates both the instantaneous collisions between
particles both temporarily and spatially
(Observation~\ref{observation:time-step-size-calculation:contact-approximation}).
Any overlap of the $\halo $-layers is considered an immediate contact
\cite{kane1985dynamics} and feeds into a hard ``sphere'' model with a perfectly
ellastic collision, which we can damp due to material laws.
Let $\mathbb{C}$ denote all contact points within a cluster, and let $p_1(c),
p_2(c) $ identify the two particles associated with a contact
point $c \in \mathbb{C}$.
We end up with a constrained equation system for the exchanged impulse 

\begin{eqnarray*}
 \forall c \in \mathbb{C}: \qquad 
 m(p_1(c))a + m(p_2(c))a
 & = & 0,
  \\
 \text{s.t. }  
 \forall p_1, p_2 \in \particleSet_i, p_1 \not= p_2: \quad
 \quad v(p_1(c),n(c)) - v(p_2(c),n(c)) & \geq & 0,
\end{eqnarray*}

\noindent
where $n(c)$ is the normal of the contact point $c$, $m$ denotes masses of the
involved particles, $a$ the acceleration and $v$ the particles' velocity.
We solve for $a$.   
The arising $v$s depend on the accelerations from all the contact points plus
the incoming velocity as well as the time step size
chosen.

The constraint phrased over the outgoing velocity makes the scheme implicit.
Due to it, particles stop approaching immediately as soon as they get in
contact.
As our time step size is fixed at this
point, we can disregard the mesh, argue solely about the set of identified contact points, and consider each particle to
be represented by a current and future transformation resulting from this set of
contact points.

This is a non-trivial, non-linear equilibrium equation system due to the
constraints.
We apply a Picard iteration to solve it, where we update the individual
contributions per contact point one by one.
The constraints are weakly enforced via a penalty term, and we
damp the iterations' updates using linear combinations of previous and newly
computed solution guesses.
In line with other codes, we assume that the Picard iterations converge
\cite{Schmitt05_17}, which is reasonable as we minimise
$\snapshotCollision$ such that only few contact points remain.

%


Since the constraints are phrased over relative velocities and not on the
lengths of the contact normals, we can end up in a situation where particles
gradually drift towards each other over several time steps, even though
their relative velocity is canceled each time they come into contact.
They might be pressed closer and closer together by other particles in the
system.
Eventually, they might even start to penetrate.
This happens notably once large stacks of objects settle into a pile of
particles on the ground.

Therefore, we add an additional force term to slowly
separate the particles.
This artificial term does not contribute to the resulting velocities.
Further to that, we use a simple Coulomb friction model where an impulse is
applied along a tangent to the contact normal to reduce the sliding velocity.
This impulse is bounded by the impulse along the normal multiplied by the
coefficient of friction between the surfaces.

\paragraph{Efficiency, contextualisation and implementation.}

An alternative, popular force model replaces the Dirac force distribution with a
smooth force function.
Its support is bounded by $\halo$, i.e.~it is zero if the distance between two 
objects or triangles exceeds $2\halo$, i.e.~their contact point $c$ carries
$|n(c)| \geq \halo$.
Finally, the contact force $\lim _{|n(c)| \mapsto 0} f(|n(c)|) = \infty$, as the
distance becomes smaller and smaller.
The smooth approximation of the force function means that particles become
weakly compressible, and their interaction resembles the compression of a
spring.
We employ this model in our previous work \cite{multires2022}.

The present approach also can be read as an approximation of the force.
However, instead of embedding a smooth force function into the $\halo$-area, we employ a
heavyside function which jumps to a constant value at the point that we identify
as collision time stamp.
It is the magnitude of this function which dynamically adapts to the
state of all particles.
It is hence an semi-implicit formalism of particle collisions:
the position updates continue to be handled explicitly, the arising forces aka
accelerations are integrated implicitly.
As we solve the arising equation system (and add an artificial spring term), we
have a conservative physical model which ensure that we never run into
unphysical solutions.
Overall, this part of the scheme is similar to an event-driven
simulation~\cite{Luding:2008:DEMIntro}.

The whole contact resolution is non-trivial to parallelise effectively
without fundamental changes to the algorithm.
Locks or atomics for example allow us to update all particles involved in
contacts in parallel within the Picard iteration.
Our experiments with these techniques however showed a drastic decreases in performance.
We instead colour the per-cluster graph $c$ and update those particle's momentum
in parallel which do not simultaneously feed into other updated particles. 
That is, if there is a collision between $p_1$ and $p_2$ and $p_2$ and $p_3$,
$p_1$ and $p_3$ are updated in one rush, while $p_2$ is updated afterwards.
This bounds the concurrency of the implementation yet outperformed alternative
trials.

While we note that clusters can still be updated in parallel in this phase, we
note that the cluster masking in
(\ref{equation:cluster-masking:min-collision-time}) effectively disables most
clusters from  processing.
The concurrency level resulting from multiple cluster is limited.

\section{Efficacy and correctness}
\label{section:correctness}

The proof of Theorem \ref{theorem:simulation-valid} is technical yet brief
with a combination of an induction over the number of particles plus a
contradiction.
As we terminate our algorithm as soon as the global time stamp $\snapshotCurrent$
exceeds the terminal time, the theorem states that the simulation outcome is
correct.

\begin{proof}
 The correctness is trivially true for one particle.

 Let all state transitions be globally enumerated.
 We use the counter $n$ for the enumeration. 
 Assume that the tuple
 $(\snapshotOld _n, \snapshotCurrent _n)(p_i)$ produced for particle $p_i$ in
 step $n$ is invalid:
 it should not have been made, as we dropped an old solution in this step
 to which we should have rolled back.
 This formalism implies that step $n$ triggered the state transition
 \[
   (\snapshotOld _{n-1}, \snapshotCurrent _{n-1})(p_i)
   \mapsto
   (\snapshotOld _n, \snapshotCurrent _n)(p_i)
   \quad
   \text{with}
   \
   \snapshotCurrent _{n-1}(p_i) = \snapshotOld _n(p_i).
 \]
  
 \noindent
 By induction, we know that $1 \leq i \leq |\particleSet|-1$ and that it has
 been particle $p_{|\particleSet|}$ which would have crashed into 
 $p_i$ at a time $\snapshotCollision < \snapshotOld _n$.
 This is the collision we missed.

 $p_i$ and $p_{|\particleSet|}$ have been in different clusters in step $n$.
 Otherwise, they would have advanced in sync.
 Therefore, there is a step 
 $m<n$ in which $p_{|\particleSet|}$ has updated its position the last
 time.
 \vspace{-0.15cm}
 \[
 \snapshotCurrent _m (p_{|\particleSet|}) = \snapshotCurrent _{m+1} (p_{|\particleSet|}) = \ldots =
 \snapshotCurrent _{n-1} (p_{|\particleSet|}) = \snapshotCurrent _{n} (p_{|\particleSet|})
 \]
 
 \noindent
 from thereon, i.e.~we haven't updated this particle anymore since then.
 According to our assumption, we learn about this missing collision after step
 $n$, but the actual collision happened at
 $\snapshotCurrent _{m} (p_{|\particleSet|}) <
 \snapshotCollision < \snapshotOld _n (p_i)$.
 Now we cannot roll back anymore.
 This can happens if and only if particle $p_{|\particleSet|}$ lacks behind.   
 
 $p_i$'s transition in step $n$ is allowed if and only if the other particles in
 their clusters have not been left behind.
 Otherwise, we violate (\ref{equation:cluster-masking:min-collision-time}).
 This is a contradiction to the time constraints on $\snapshotCollision $ above.
 \vspace{-0.5cm}
 \begin{flushright}
 $\qed$
 \end{flushright}
\end{proof}

\noindent
Theorem~\ref{theorem:simulation-terminates} rephrases that
$\timeMinCurrent \geq \timeTerminal$ after a finite number of time steps.
With a detailed description of the individual algorithmic steps, we can 
indeed show that our algorithm terminates.

\begin{lemma}
 The time step sizes of active particles are significantly bigger than zero,
 i.e.~they do not deteriorate.
\end{lemma}

\noindent
The lemma can only be proven once we assume that the kinetic energy of a system
is not growing, i.e.~as long as the system is closed and no energy is
injected.

\begin{proof}
 Let a particle that is not subject to changes of its kinetic energy
 follow a trajectory of straight line segments $l_0, l_1, l_2, \ldots $.
 This is the analytic trajectory, i.e.~no two line segments point into
 the same direction.
 To contradict the lemma, $|l_n| < |l_{n-1}|$.
 Such a situation can arise if a particle, for example, falls into a funnel or
 hopper and bounces in-between the walls coming closer and closer, until the
 particle finally gets stuck in-between walls.
 
 We first recognise that the reduction above means $\lim _n |l_n| = 0$, as our
 particle interaction model is subject to friction. 
 A particle notably cannot bounce forth and back between two walls without
 slowing down. 

 However, to construct such a trajectory is not possible for our algorithm, as
 the particles are surrounded by an $\halo $-layer. Hence, the
 segment length is bounded by $|l_n \geq \halo$ or two particles are in a
 ``permanent'' contact, i.e.~stick to each other. In the latter case, the
 interaction is explicitly removed from the time step size calculation.
 \vspace{-0.5cm}
 \begin{flushright}
 $\qed$
 \end{flushright}
\end{proof}

\noindent
The lemma formalises the $\gg $ in (\ref{equation:time-step-size-calculation}).
It becomes clear that friction works in our favour in this context: 
As particles slow down due to friction after each collision, even constant path
segments would lead to larger time intervals in-between collisions.
Even without friction, our $\halo$-formalism effectively truncates analytical
trajectories of line segments which would approach zero.
One can still construct arbitrary long sequences of time step sizes
approaching zero if a system is permanently stimulated by injecting additional kinetic
energy.
We neglect such setups.

\begin{observation}
  The semi-implicit collision model, where the contacts are prescribed and only
  impulses are implicitly determined avoids time step size degeneration and
  bouncing particles when we encounter compact, settled particle geometries.
\end{observation}

\noindent
With a smooth spring force, time step sizes have to be dramatically reduced as long as two objects
continue to approach each other despite the actions of the force.
As a consequence, global assemblies of objects with multiple contact points will
rarely reach an exact equilibrium. 
They will always oscillate around the steady state solution as long as an
external force such as gravity holds them together or otherwise disassemble.
While a force-based formalism may have its advantages, such vibrating assemblies
of objects are problematic for local time stepping, as they lead to an effective
time step size degradation. 
It is hence not only an algorithmic decision to determine the contact points and
impulses implicitly each---though the overall schemem remains explicit as these
to ingredients are not coupled---but it is essential to make the local time
stepping work.

\begin{lemma}
 The particle with the smallest $\snapshotCurrent $ is always a member of an
 active cluster.
\end{lemma}

\begin{proof}
 We assume the lemma were wrong an construct a contradiction:
 Let $\snapshotCurrent(p_i) = \timeMinCurrent$, i.e.~let $p_i$ determine the
 global minimum time stamp. 
 Another particle $p_j$ carries $\snapshotCurrent(p_j)>\snapshotCurrent(p_i)$.
 Let $p_i$ and $p_j$ end up in different clusters, were they both determine the
 minimal cluster time stamp $\timeMinCurrent (\particleSet _i)$ or
 $\timeMinCurrent (\particleSet _j)$ respectively.
 
 If $p_j$ vetos the advance of the $\particleSet _i$, 
 \[
 \snapshotCurrent (\particleSet _i) > \snapshotCollision.
 \]
 However, $\snapshotCollision (p_j)>\snapshotCurrent(p_j)$ by definition.
 \vspace{-0.5cm}
 \begin{flushright}
 $\qed$
 \end{flushright}
\end{proof}

\noindent
What could happen however is that cluster $\particleSet _i$ takes a large time
step and overtakes the particles within $\particleSet _j$, missing out on
collisions at $t>\snapshotCollision(p_j)$. In this case, it will roll back later
to $\snapshotCollision(p_j)$ as particles from $\particleSet _i$ and
$\particleSet _j$ are joined into one cluster.

\begin{lemma}
 Every particle is updated after a finite number of steps, i.e.~every particle
 becomes a member of an active cluster.
\end{lemma}

\begin{proof}
 The lemma (and proof of Theorem~\ref{theorem:simulation-terminates}) results
 directly from the previous lemmas: 
 A particle $p_i$ is not strictly advancing in time. 
 However, let $T=\sum _{p_j} (\snapshotCurrent (p_i)-\snapshotCurrent (p_j))$ be
 a sum over all particles $p_j$ with $\snapshotCurrent (p_i) > \snapshotCurrent
 (p_j)$. This sum is strictly monotonously decreasing with a decrease that
 bounded away from zero.
 After a finite number of steps, particle $p_i$ is therefore updated.
 \vspace{-0.5cm}
 \begin{flushright}
 $\qed$
 \end{flushright}
\end{proof}

\section{Runtime results}
\label{section:results}

All experiments were
run on AMD EPYC 7702 chips in a two socket configuration with $2 \times 64$ cores.
Each socket is divided into 4 NUMA regions.
%
Our code was translated with the Intel oneAPI C++ Compiler \texttt{icpx}
2021.4.0 using the flags \texttt{-std=c++17 -O3 -march=native},
i.e.~we tailored it to the native instruction set.
Our realisation relies on 
CGAL 5.5.1 \cite{CGAL}, Eigen 3.4 \cite{Eigen}, Embree 3.13 \cite{Embree}, libigl
2.3 \cite{libigl}, SIMDe 0.7.2 \cite{SIMDe}, TetGen 1.6.0 \cite{TetGen}, Catch2 3.1.0
\cite{Catch2} and GLFW 3.3.6 \cite{GLFW}.
As our prime interest is runtime, we hardcode key geometric properties:
$\epsilon = 10^{-2}$ is uniformly used unless specified differently, and 
the Picard solver for the particle interaction steps considers an
estimate of the momentum exchange as converged when the
update to the impulse and torque applied to every contact underruns a
threshold of $10^{-4}$.

We always benchmark our local time stepping against a global approach. The
latter uses the same software architecture, i.e.~works with clusters, an implicit momentum
exchange and per-cluster collision timestamp identification, too.
To make it a globally adaptive scheme,
we use 
(\ref{equation:cluster-masking:min-collision-time}) to constrain all time step
sizes:
Clusters identify their effective time step size, but then reduce  
it further whenever another cluster identifies that it can only advance with a
smaller one.
This mirrors textbook adaptive global time stepping.
Due to the damping of the time step size in
(\ref{equation:cluster-step:damped-particle-time-step-size}), a locally stiff
collision implies that all particles in the stimulation adopt a tiny time step
size.
Besides this additional constraint---which can be read as an explicit further
narrowing of the time step size per cluster once the time step size calculation
has terminated (cmp.~Definition~\ref{definition:narrowing})---the global time
stepping code equals excactly the local version.
Notably, it exploits the same parallelisation techniques and hence provides a
fair comparison baseline.

\subsection{Particle pairs}

\begin{figure}[htb]
 \begin{center}
  \includegraphics[width=0.24\textwidth]{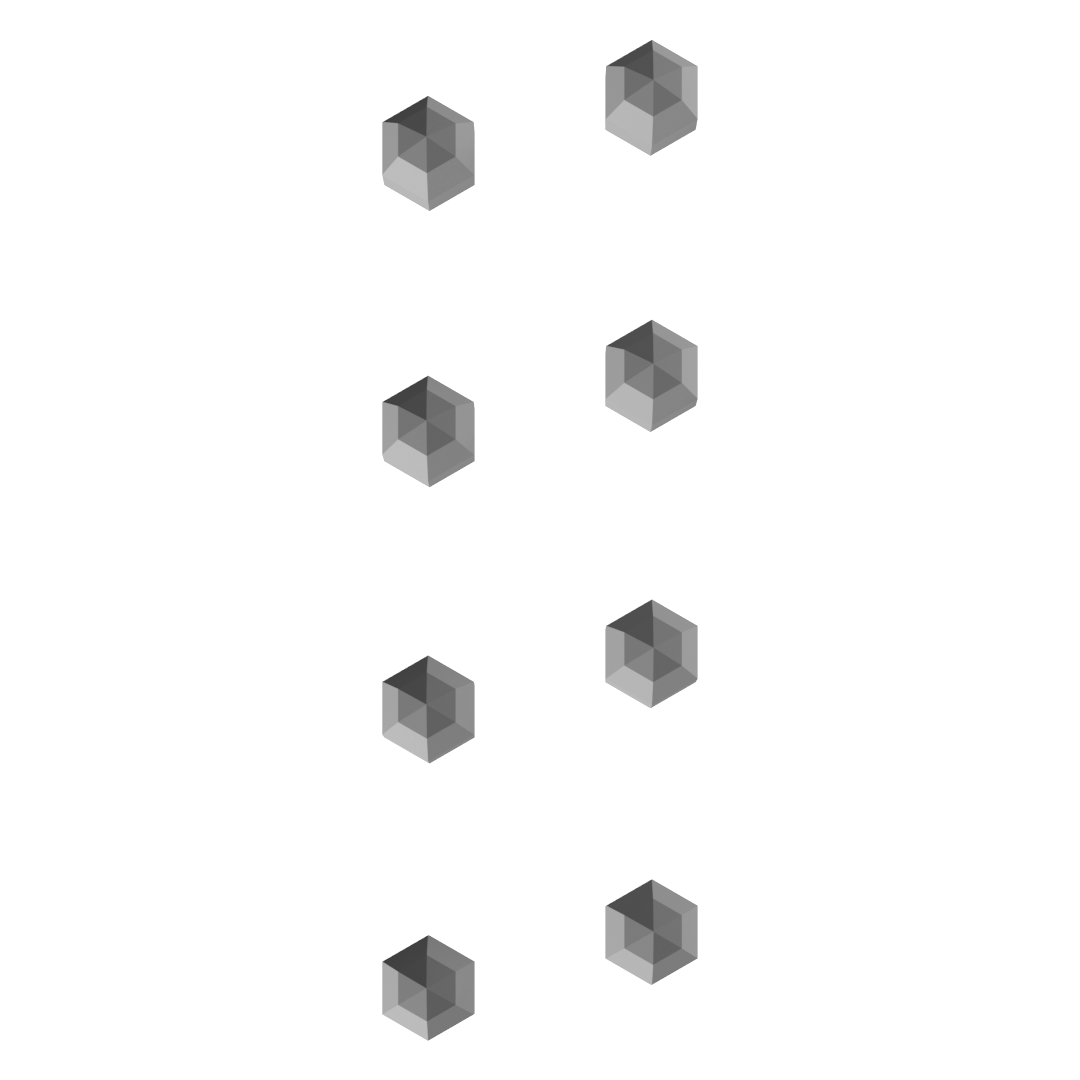}
  \includegraphics[width=0.24\textwidth]{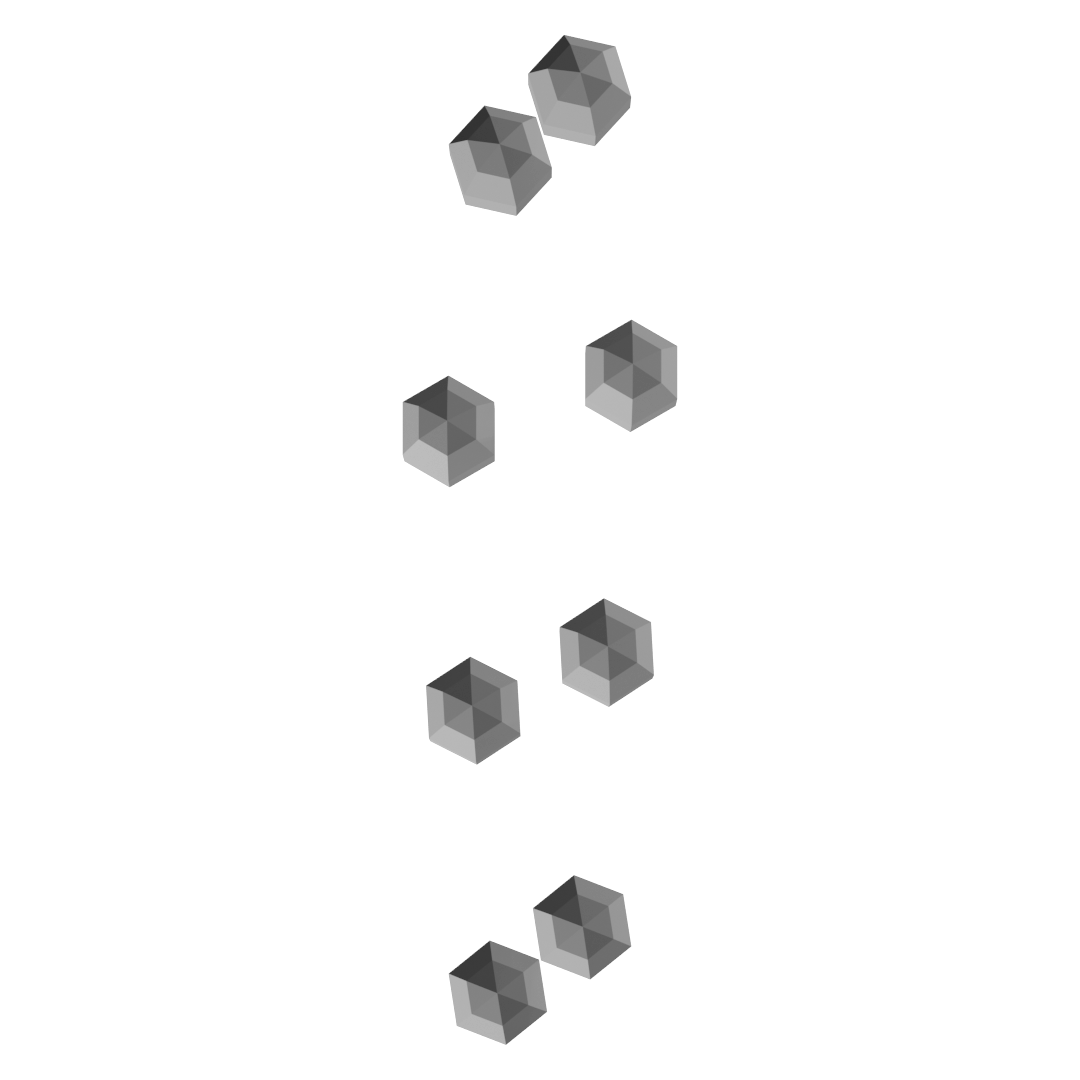}
  \includegraphics[width=0.24\textwidth]{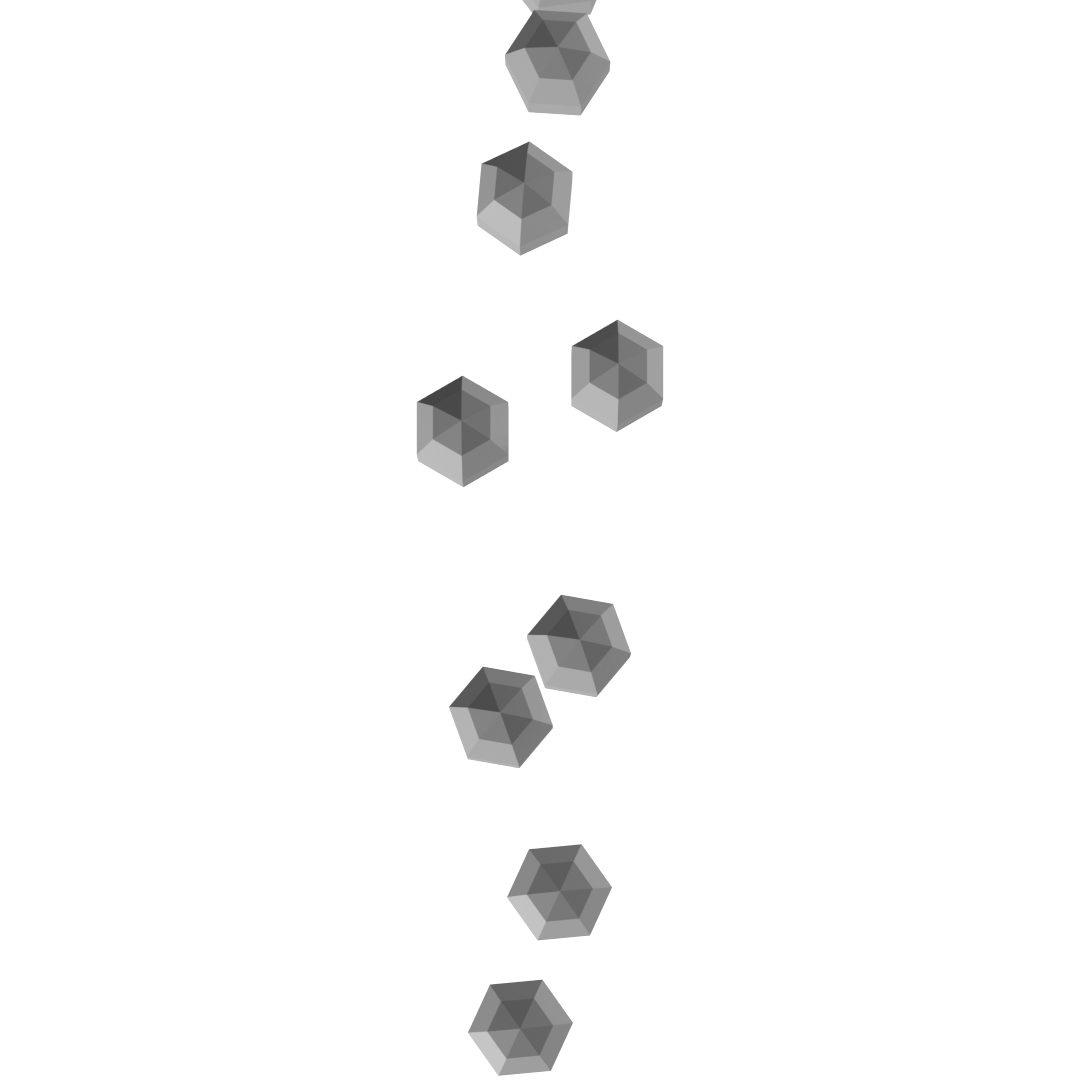}
  \includegraphics[width=0.24\textwidth]{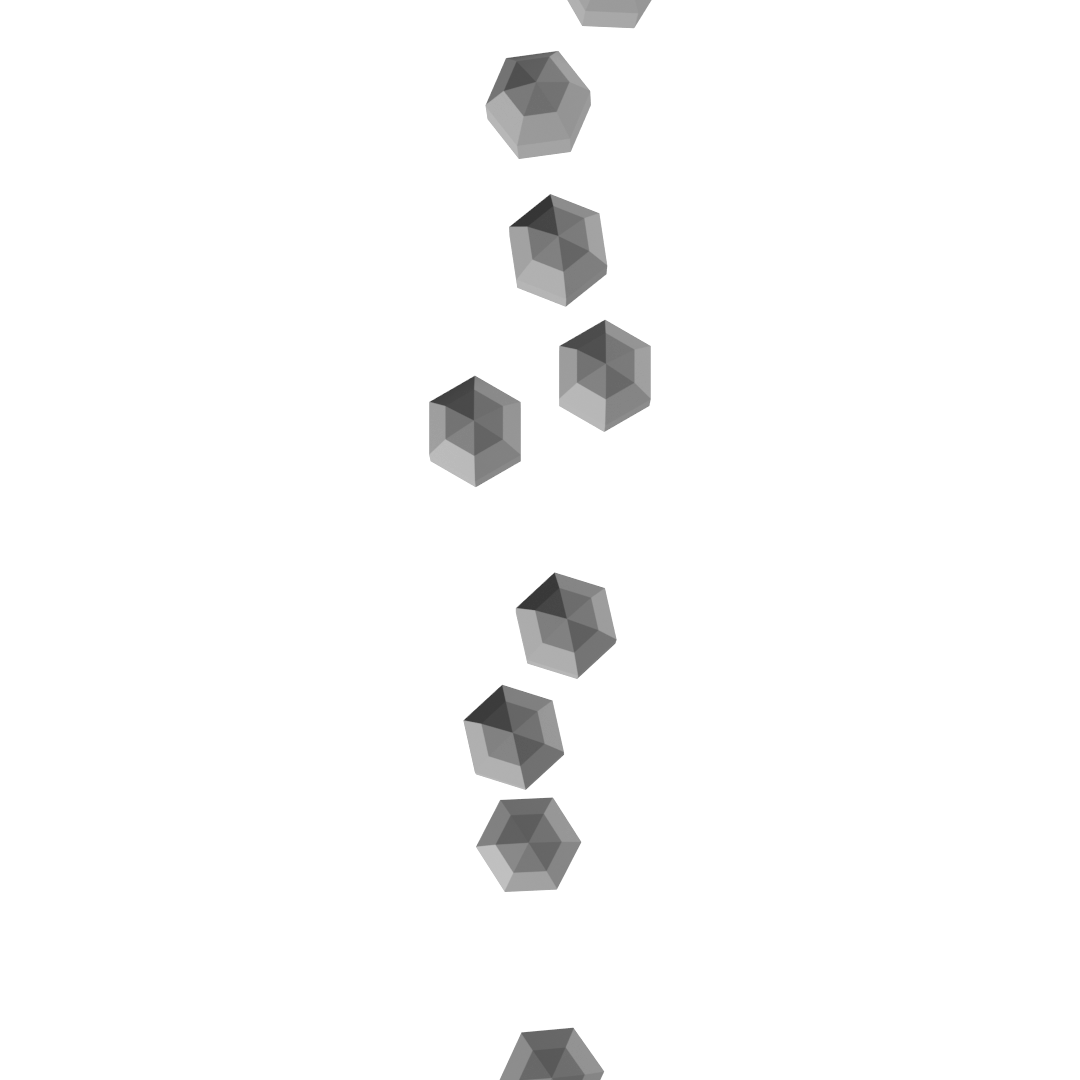}
 \end{center}
 \caption{
  The particle-pairs scenario from left to right:
  Starting from a column arrangement, pairs of particles
  (cubes in this illustration) approach each other with randomised velocity.
  The pairs collide and separate again.
  \label{figure:results:particle-pairs:setup}
 }
\end{figure}

In a first experiment, we create two sets of particles.
Both host the same number of
particles, which are arranged vertically yet spaced out.
They do not touch
each other (Figure~\ref{figure:results:particle-pairs:setup}).
The particles used are of identical size, spherical, and we use 
$|\mathbb{T}|=48$ and $\epsilon = 0.01$.

The two sets are put next to each other like columns.
The two columns in turn are set on a horizontal collision trajectory,
such that always two particles bump into each other.
As we randomly vary the initial horizontal velocity magnitude per particle, the
collision time stamps differ.
We configure the initial geometric arrangement of this
particle-pairs scenario such that each particle collides exactly once with its counterpart.

\begin{figure}[htb]
 \begin{center}
  \includegraphics[width=0.49\textwidth]{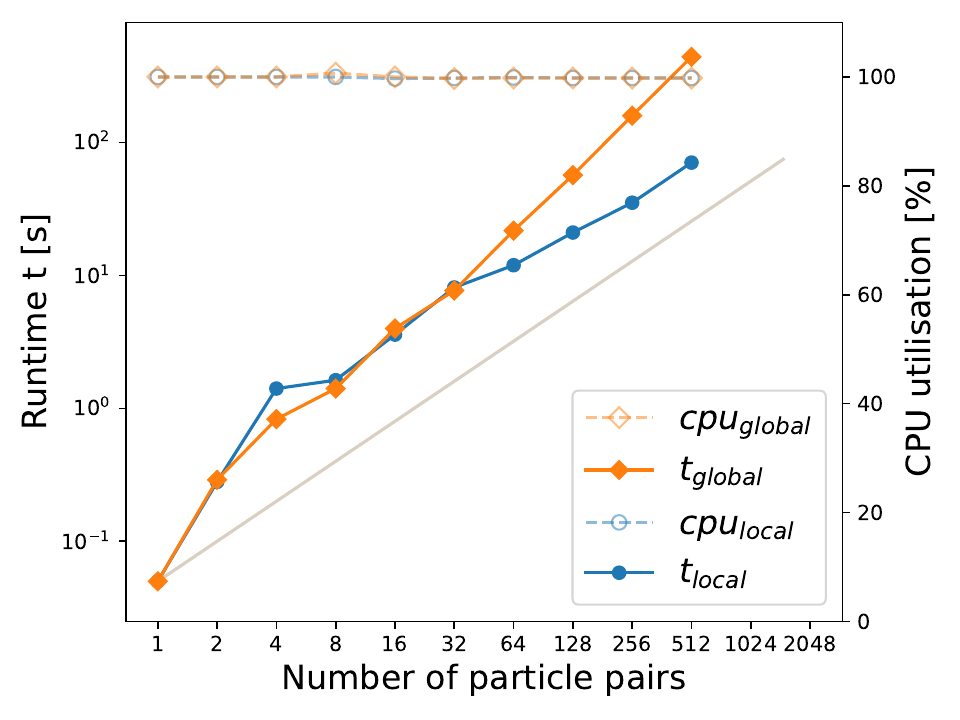}
  \includegraphics[width=0.49\textwidth]{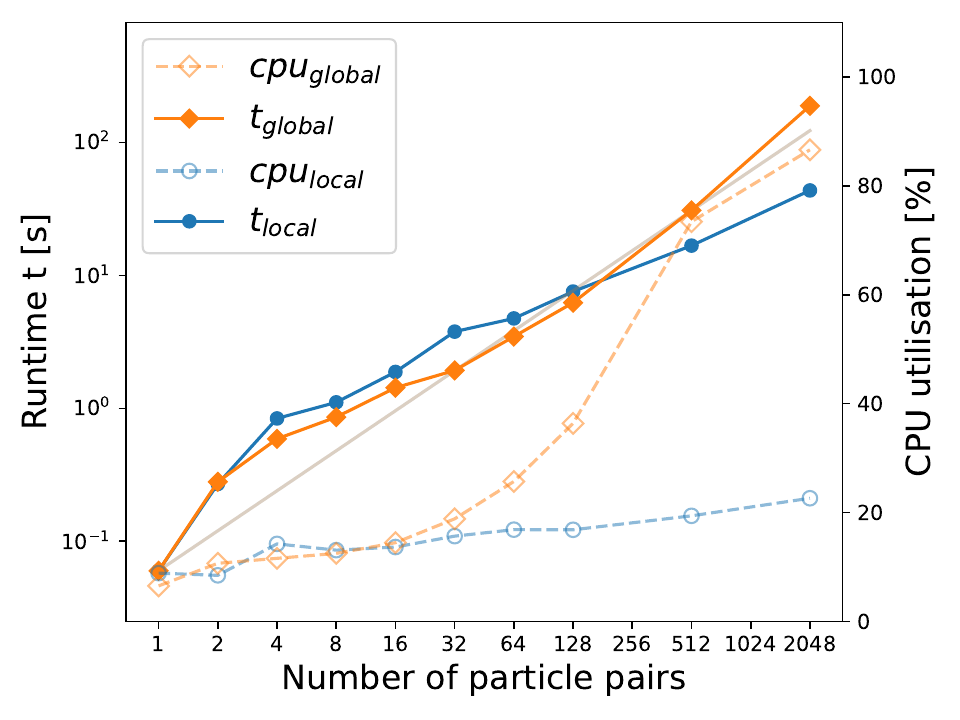}
 \end{center}
 \caption{
 Runtime and CPU utilisation for the collision of particle pairs on a single
 core (left) and the whole node (right). 
 The CPU utilisation is normalised against the number of cores available.
  \label{figure:particle-pairs:runtime}
 }
\end{figure}

%
%
Each particle pair runs through three computational phases:
While the particles approach, there is no collision and no forces act on the
particles, as we omit gravity.
When they collide, the system becomes
very stiff suddenly and the particles exchange forces.
When they separate again, there are no forces.
As we use random initial horizontal velocities, the phases of the different
particle pairs are not synchronised in any way.

\begin{observation}
 On both a single core and on multiple cores, local time stepping outperforms global time stepping for
 the particle-pair setup as soon as we accommodate sufficiently many particles. 
 It scales close to linear with the number of particles,
 while it underutilises the available CPU resources.
\end{observation}

\paragraph{Cluster topology, active clusters and concurrency level.}
%
%
Tracing of the algorithm's internal states provides insight into the
runtime behaviour (Figure~\ref{figure:particle-pairs:runtime}):
The algorithm observes a maximum cluster size of up to five particles.
This is up to 2.5 times worse than the theoretical optimum, as we know that 
at most two particles at a time should be member of one cluster.
As the particle pairs are totally out-of-sync, the clusters ``diverge'' in time
quickly.
Due to the global constraint (\ref{equation:cluster-masking:min-collision-time}), this leads to a situation
where only few clusters can be updated per time step.
We obtain a low CPU usage (below 25\%).

\begin{figure}[htb]
    \centering
    \begin{subfigure}{.39\linewidth}
        \centering
        \includegraphics[width=\textwidth]{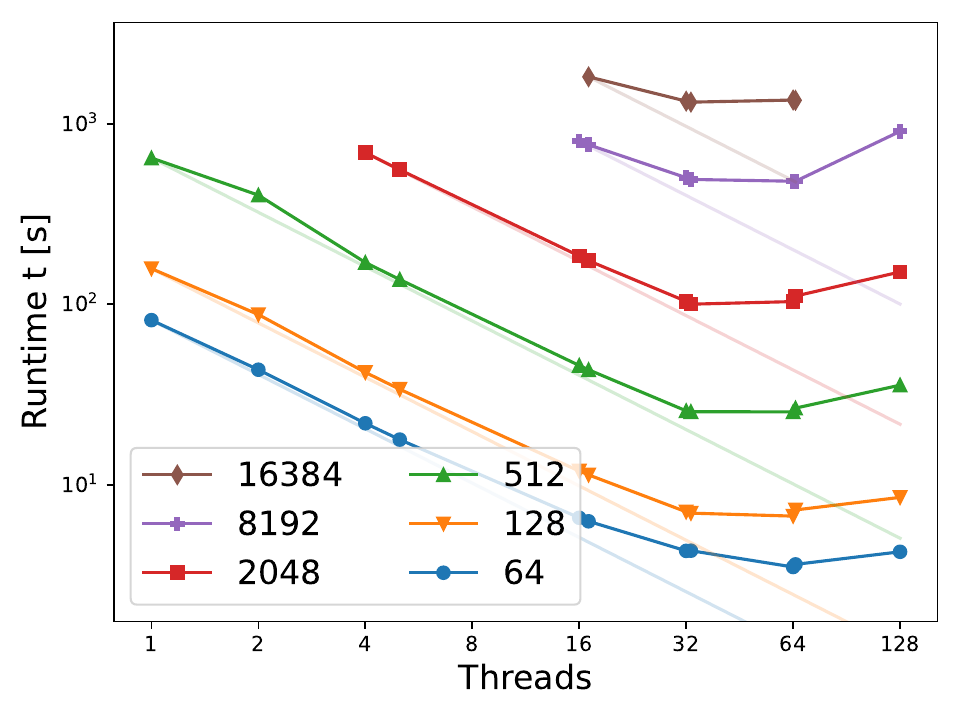}
    \end{subfigure}
    \hfill
    \begin{subfigure}{.6\linewidth}
     \tiny
        \centering
        \begin{tabular}{lccc}
            \textbf{Phase} & \textbf{Run time} & \textbf{Run time} & \textbf{Run time}\\
             & \textbf{64 pairs} & \textbf{512 pairs} & \textbf{8192 pairs}\\
            Broad collision & 0.9\% & 0.5\% & 1.0\% \\
            Clustering & 0.2\% & 0.2\% & 0.7\%\\
            Time step size & 82.6\% & 83.0\% & 82.7\% \\
            Narrow collision & 16.2\% & 16.0\% & 15.0\% \\
            Collision & 0.2\% & 0.2\% & 0.6\% \\
            Snapshot roll-over & 0.0\% & 0.1\% & 0.0\% \\
        \end{tabular}
        \begin{minipage}{.1cm}
        \vfill
        \end{minipage}
    \end{subfigure} 
    \caption{
    Left:
    Various strong scaling curves for the particle-pairs setup.
    Different numbers of particles are used with the local time stepping scheme. 
    Right: Break down of the different algorithm phases.
  \label{figure:particle-pairs:scaling}
  }
\end{figure}

Despite the low CPU utilisation, we
get reasonable weak scaling (Figure~\ref{figure:particle-pairs:scaling}).
Since we know that the clusters all hold similar particle counts (usually two),
we have to assume that there are few active clusters per time step, yet usually
more clusters than we have numbers of cores. 
Furthermore, the data suggests that the number of active clusters per time step
scales almost linearly---the weak scaling is not perfect---with the number of
total particles.

\paragraph{Time step size distribution.}
We initialise our simulation with a global $\Delta t$ threshold which is
equivalent to the mean time span up to the first collision in the system.
For the global time stepping, at least one particle pair will collide slightly
earlier than that, i.e.~we immediately reduce the effective $\Delta t$ to a
fraction of the original choice.
The next few steps reduce the time step size further, as more and more
particles crash into each other (Figure~\ref{figure:particle-pairs:time-step-size}).

Local time stepping allows some clusters to advance with the maximum
initial $\Delta t$.
The time step size histogram hence shows entries for the relative 
time step size of one.
After the first particle-pair has collided, other particles continue to collide
and to ``bounce back''. The time step size histogram fans out due to
(\ref{equation:cluster-step:damped-particle-time-step-size}).

Since the time step size is aggressively reduced immediately by the global time
stepping, very few narrowing is required.
The local time stepping in contrast requires narrowing.
At the same time, the number of cluster updates with the tiny
time step sizes is by magnitudes smaller than for global time stepping, and we
see a more ``spread out'' distribution of the used time step sizes.
This explains the good performance of local time stepping 
compared to its global counterpart.

\begin{figure}[htb]
 \begin{center}
  \includegraphics[width=0.45\textwidth]{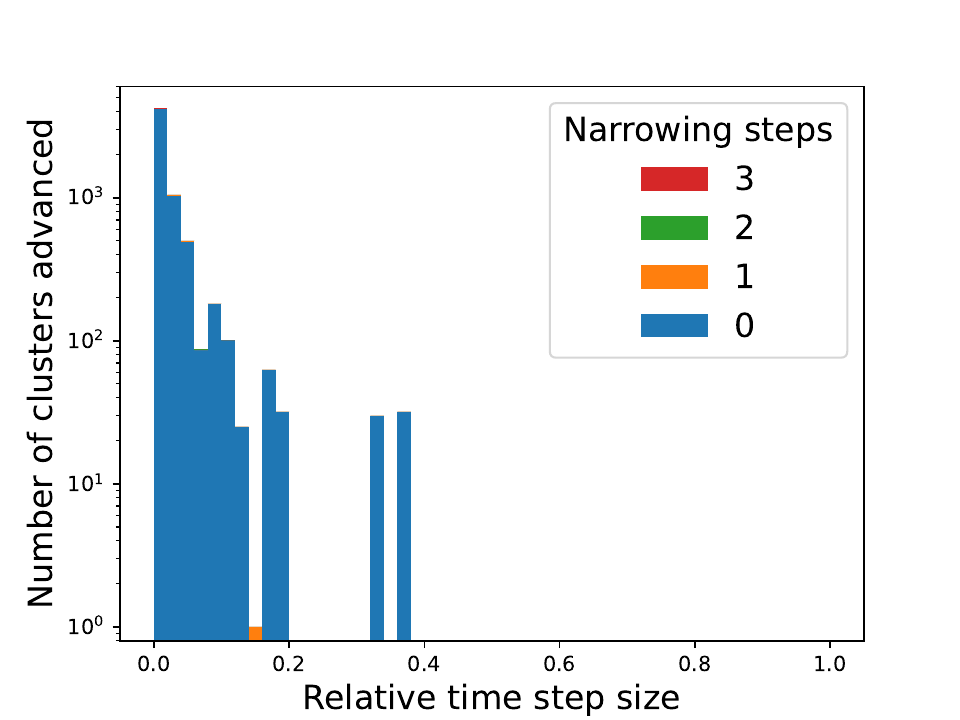}
  \includegraphics[width=0.45\textwidth]{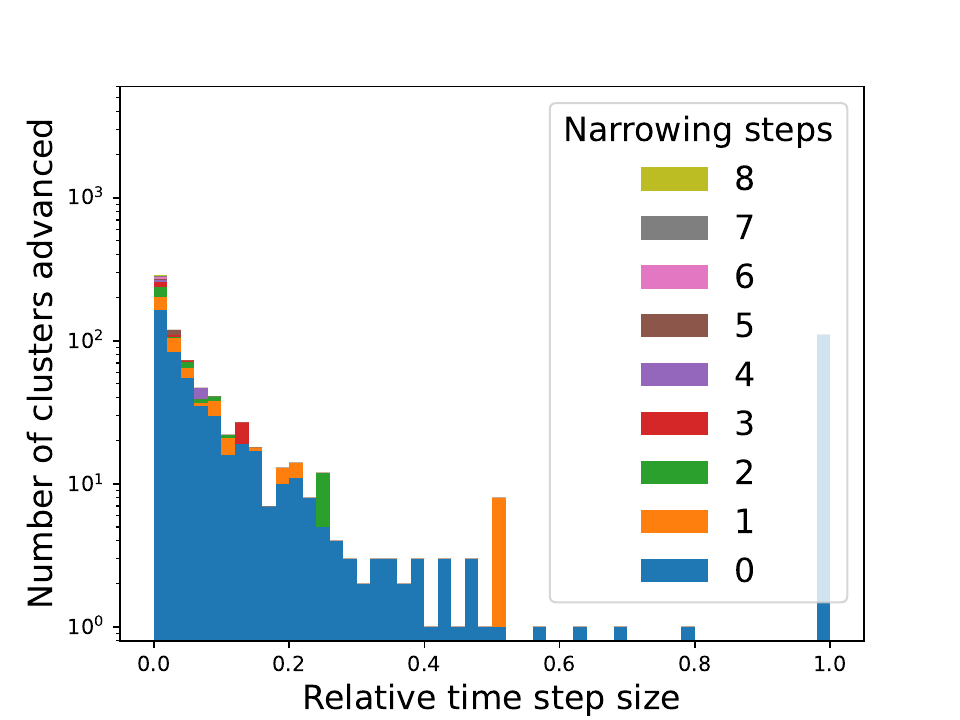}
 \end{center}
 \caption{
   Histogram of how many clusters advance in time with which time step size over
   the whole simulation of the particle-pairs scenario.
   Per chosen time step size, we also track the number of narrowing steps
   to identify an effective time step size.
   \label{figure:particle-pairs:time-step-size}
 }
\end{figure}

\paragraph{Contact points and particle interactions.}
Up to ten contacts are found per cluster, i.e.~particle-particle interaction.
Most of the time however, 
particles collide only in one single point.
The implicit momentum calculation hence becomes close to trivial.
The time step size calculation dominates the runtime
(Figure~\ref{figure:particle-pairs:scaling}).

\begin{figure}[htb]
 \begin{center}
  \includegraphics[width=0.45\textwidth]{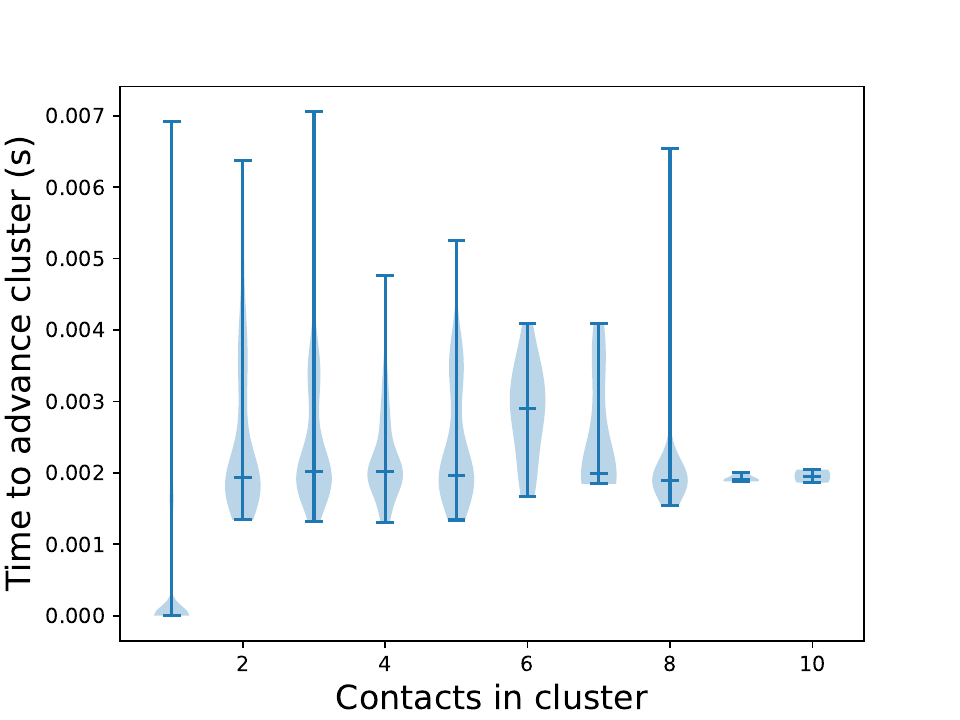}
  \includegraphics[width=0.45\textwidth]{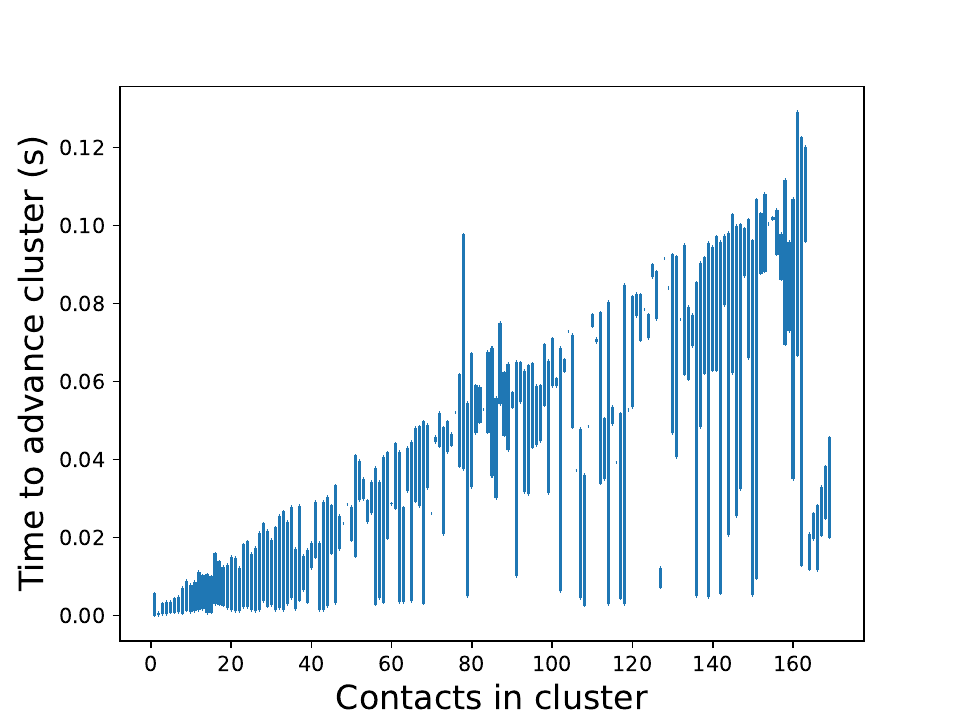}
 \end{center}
 \caption{
   Number of contacts per cluster vs.~time required to advance the cluster.
   Particle-pairs scenario (left) and particle-stack setup (right).
   \label{figure:particle-pairs:cluster-cost}
 }
\end{figure}

\paragraph{Discussion.}
For small particle counts, we obtain a superlinear cost
curve relative to the particle cardinality.
It is not clear if this is due to caching effects---as we add more particles,
they might ``fall out close caches''---a significant management overhead, or a
combination of both.
As the particle count increases, any overhead is amortised, and the local time
stepping cost curve starts to approach a linear trend.
Parallelisation subsequently helps to reduce both the curves' growth.

Our setup is of artificial character.
However, we note that many subsequent experiments lead to small particle groups
``spinning off'' from the main bulk of the objects.
As they fly away, our cluster analysis manages to find out that they barely
interact, packs them into their own cluster and handles them in parallel.
Even if a light, fast particle crashes into an assembly of compact objects, we
see that the local time stepping can handle this situation efficiently.

The setup also uncovers one further detail:
The cluster masking prohibits particles to ``shoot off'' into
the future after their first contact, since the algorithm is unaware that
particles collide at most once.
However, additional domain knowledge such as ``this particle is shooting away''
would permit us to disable the time step size constraints for the particles.
We could evolve them straight to the terminal simulation time and hence reduce
the compute cost further.

\subsection{Particle stacks and the tower setup}

\begin{figure}[htb]
 \begin{center}
  \includegraphics[width=0.24\textwidth]{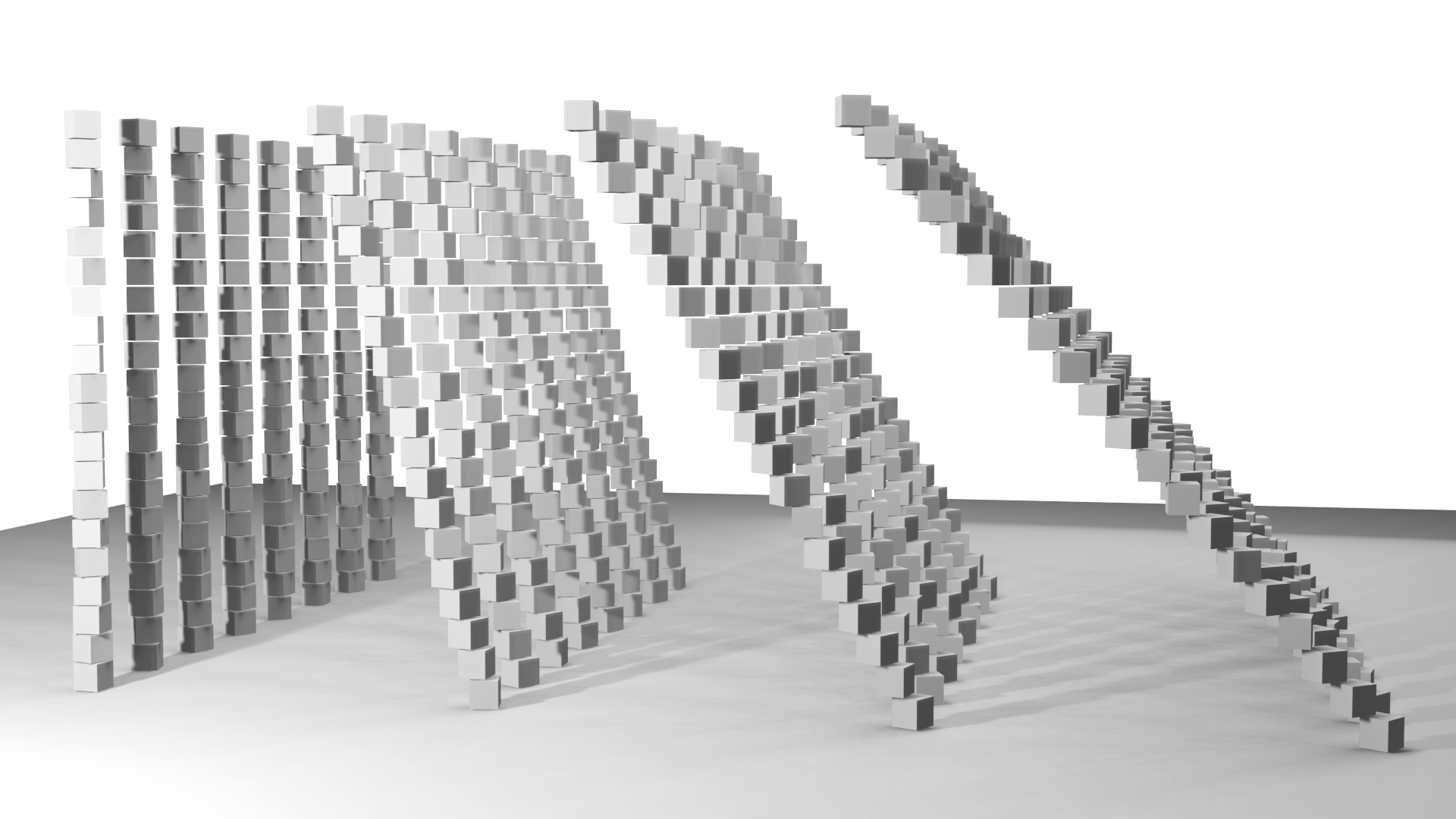}
  \includegraphics[width=0.24\textwidth]{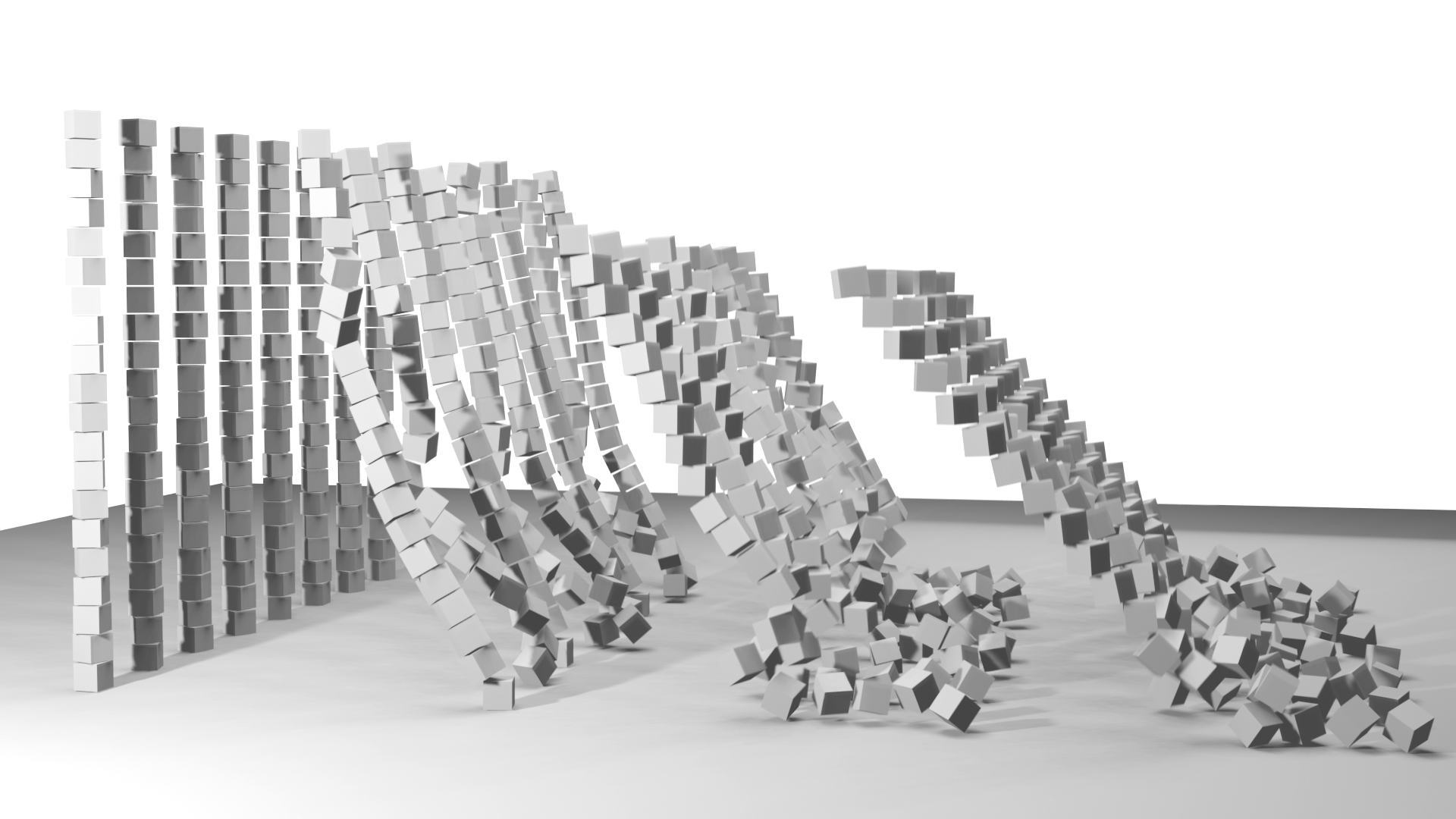}
  \includegraphics[width=0.24\textwidth]{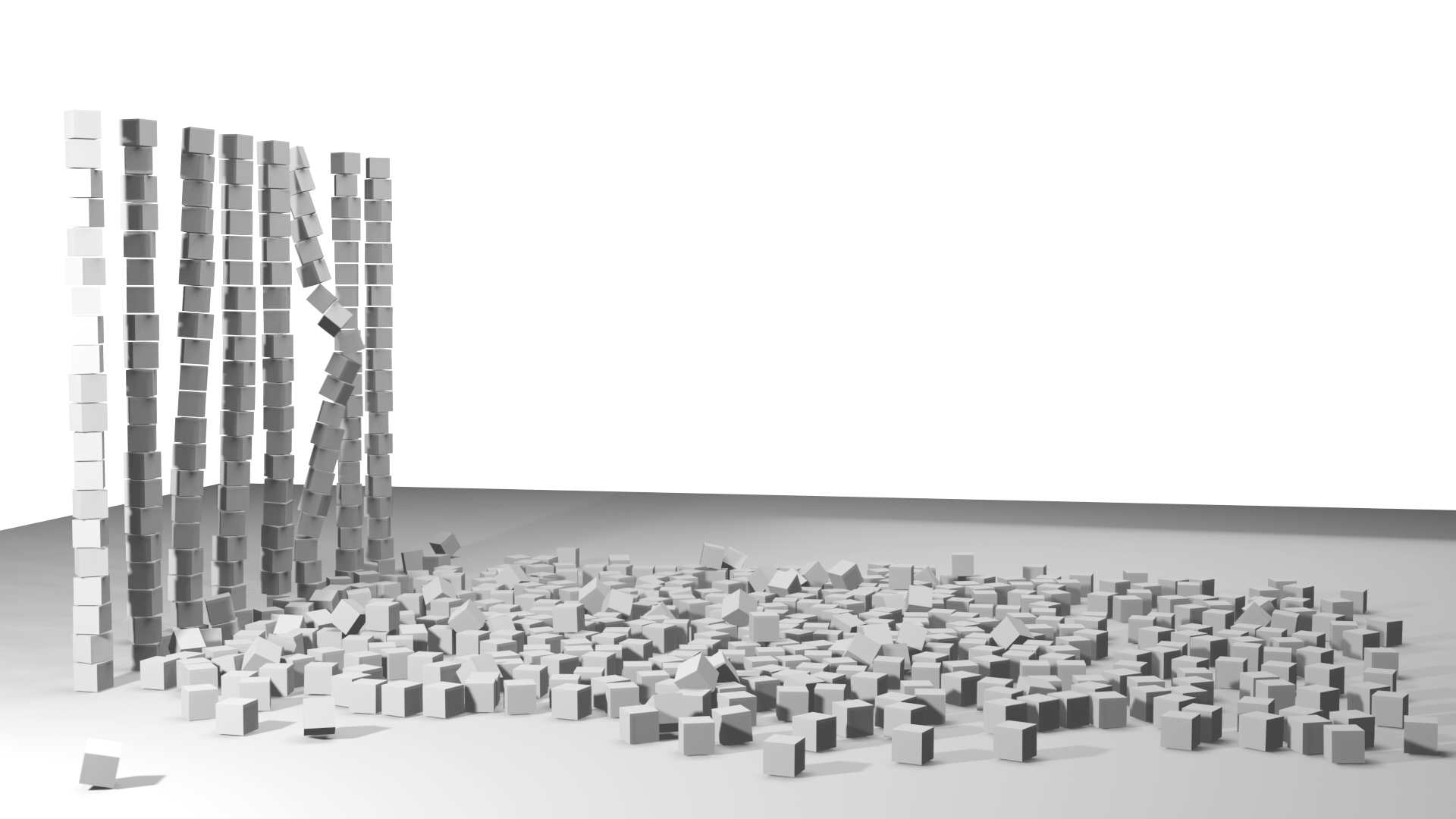}
  \includegraphics[width=0.24\textwidth]{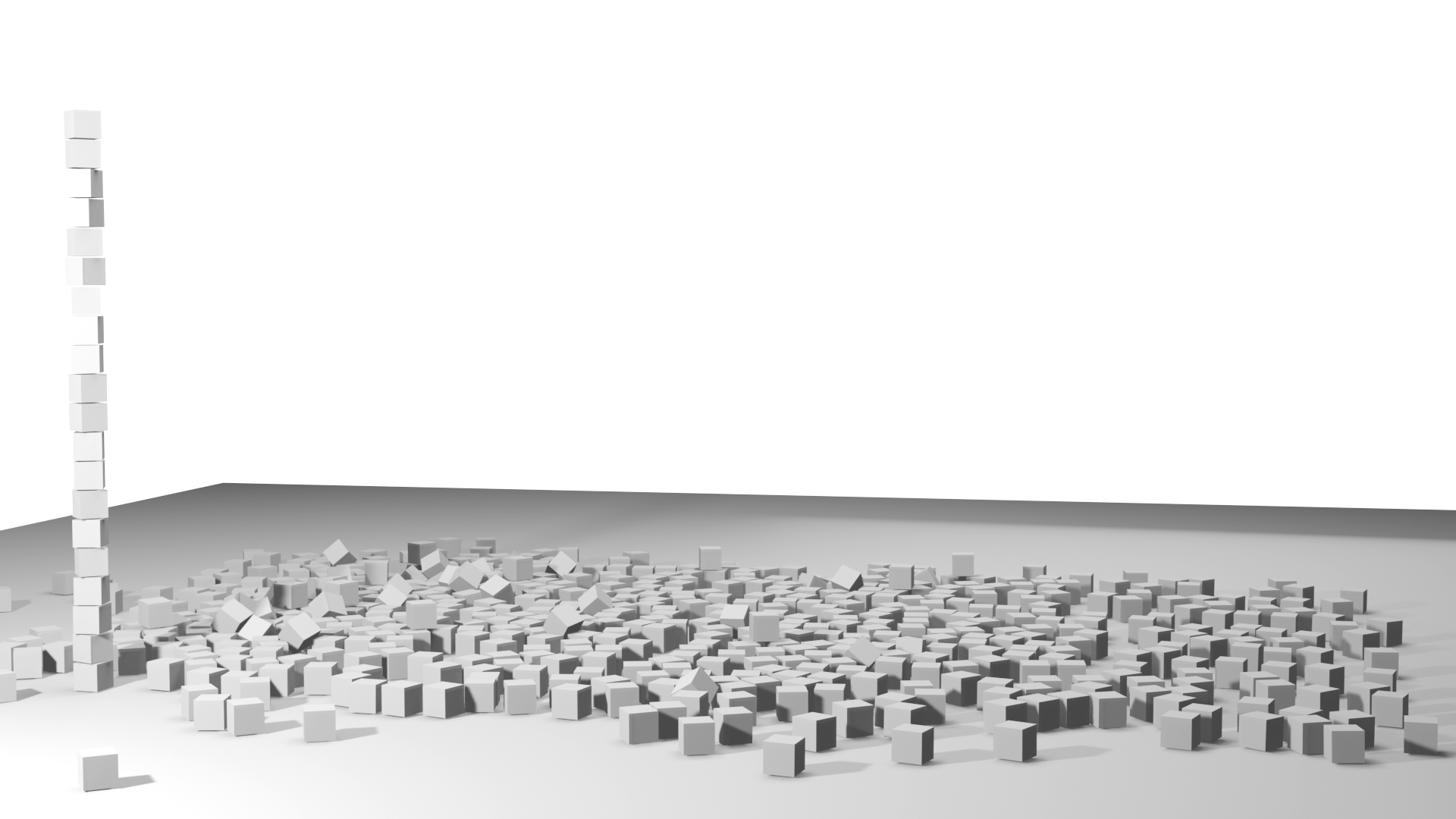}
 \end{center}
 \caption{
  The particle-stacks scenario from left to right:
  The initial state before the towers begin to topple.
  The particles from the collapsed towers disperse and some hit the stable towers causing them to topple.
  A final stationary state is reached with all particles either resting on the floor or part of a stable stack.
  \label{figure:results:particle-stacks:setup}
 }
\end{figure}

In a second setup, we study a series of
stacks of cubes which are initialised with a variety of slanting angles.
The stacks are arranged in a Cartesian layout.
As all but the left-most stacks are unstable, each stack topples over and
interacts with the other stacks (Figure~\ref{figure:results:particle-stacks:setup}).

Again, our simulation progresses through three distinct phases:
Initially, the particles aka cubes form stable clusters. 
They stiffly interact with their adjacent cubes only. 
In the second phase, cubes crash into cubes from other stacks or cubes from their own
stacks before they distribute over the floor.
In the third and final phase, the cubes ``roll'' over the floor before they
finally come to rest.

In our simulation setup, the ground plane is modelled by two $|\particleSet|=2$
triangles, and the cubes use $|\particleSet|=12$.
We scale the setup by increasing the number of rows of stacks, while we
always employ 20 cubes per stack and four stacks per row.

\begin{figure}[htb]
 \begin{center}
  \includegraphics[width=0.45\textwidth]{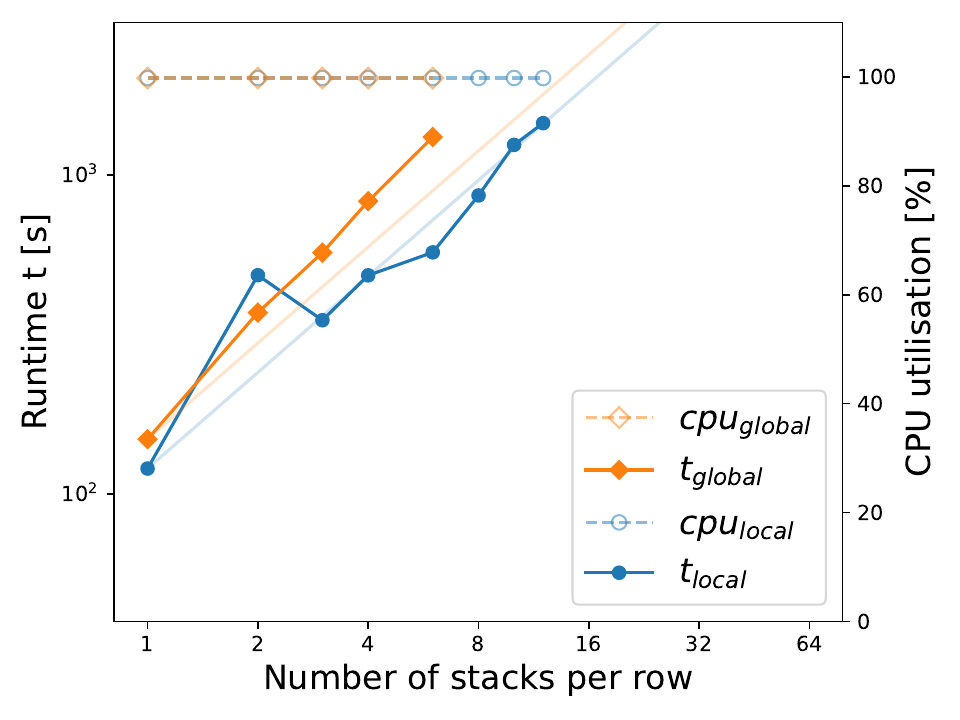}
  \includegraphics[width=0.45\textwidth]{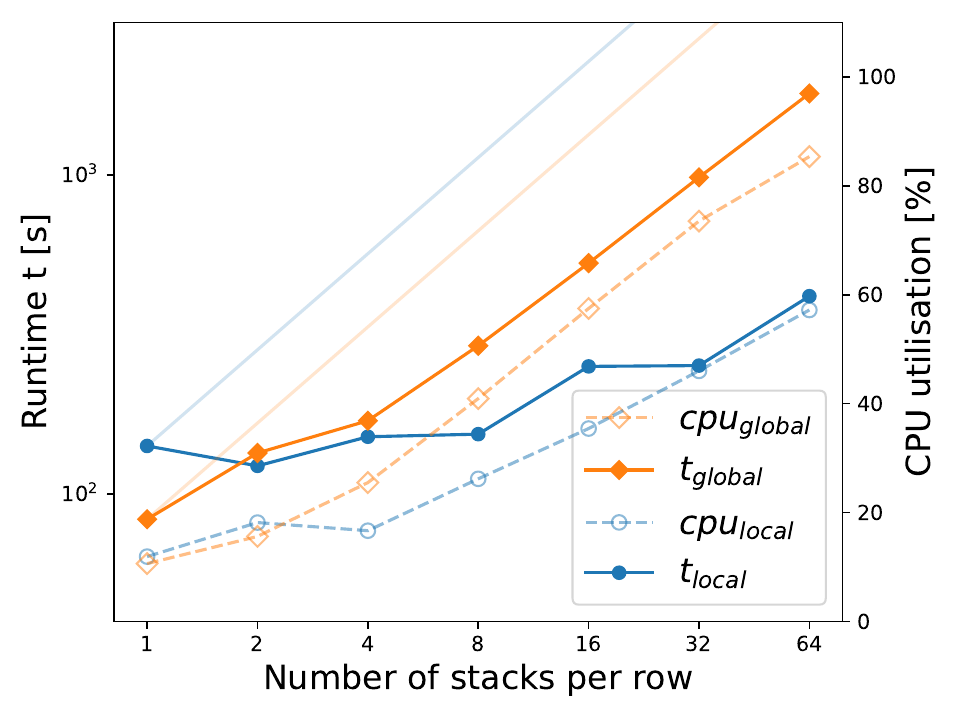}
 \end{center}
 \caption{
   Runtime and CPU utilisation for the collision of the particle/cube stacks on
   a single core (left) and the whole node (right). 
   \label{figure:particle-stacks:runtime}
 }
\end{figure}

\begin{observation}
 Except for small setups with very few stacks, local time stepping
 outperforms the global time stepping approach on a parallel computer.
\end{observation}

\paragraph{Cluster topology, active clusters and concurrency level.}

As long as the particle stacks remain more or less ``intact'', i.e.~while they
tilt, we obtain around one cluster per stack.
This is due to the fact that the stationary floor is explicitly excluded from
the cluster construction.
Once stationary, we get many clusters consisting of few
particle, as particles end up scattered over the floor.
In both situations, local time stepping manages to advance the clusters
aggressively.

In-between these two phases, the cluster topology is quite dynamic.
It is this in-between phase, where the overall system is rather stiff and
employs small time step sizes.
However, individual subclusters of particle constellations might nevertheless
advance aggressively in time; notably when they have not yet crashed into other
pillars or have already come to rest.
Overall, we see the local time stepping outperform a global time stepping once
the particle count is sufficiently high, and we even have superlinear scaling in
the number of particles (Figure~\ref{figure:particle-stacks:runtime}).

\begin{figure}[htb]
 \begin{center}
  \includegraphics[width=0.45\textwidth]{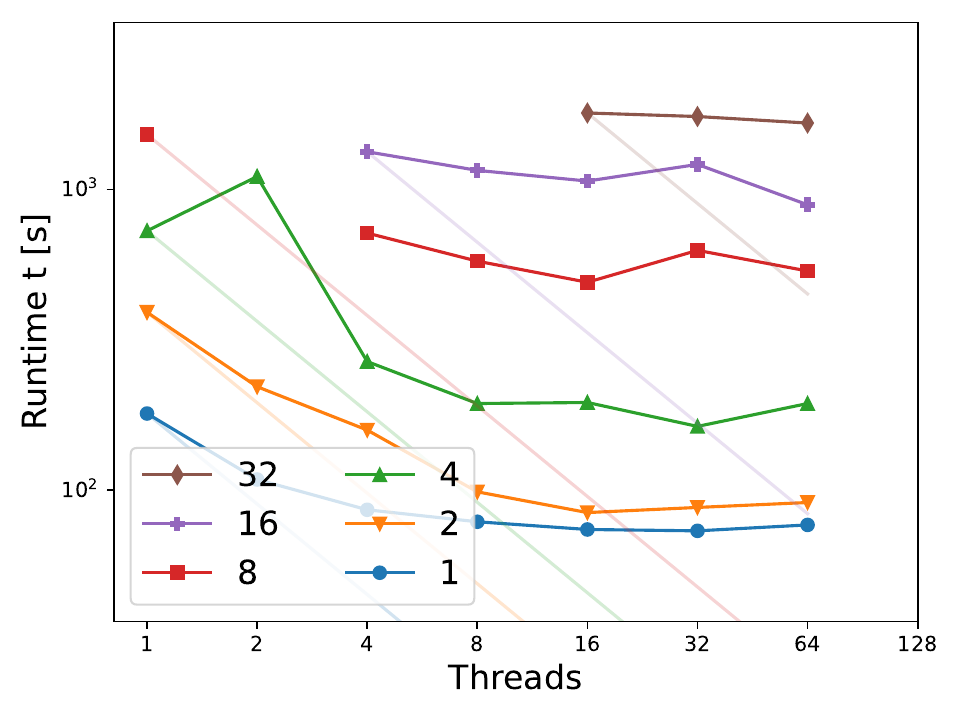}
 \end{center}
 \caption{
 Strong scaling curves for various numbers of rows of particle
 stacks.
  \label{figure:particle-stacks:scaling}
 }
\end{figure}

\paragraph{Time step size distribution.}
At the same time, this frequent change of cluster topologies implies that the
scalability is not particularly good.
This observation is independent of the total particle count
(Figure~\ref{figure:particle-stacks:scaling}).
At any point in the intermediate phase, many particles either are associated
with inactive clusters, may have to roll back in time, or may only advance with
tiny time steps.
Therefore, the algorithm struggles to benefit from additional cores.


\paragraph{Contact points and particle interactions.}
In the first and third simulation phase, particles on average experience either
one or two contact points.
The solution of the arising non-linear equation systems is rather
straightforward.
In the expensive middle phase, some particles face many stiff 
contacts within a single time span.
The arising equation systems become larger and
(potentially) require more iterations.
We end up with a linear correlation trend between the number of contact 
points and the cost to solve the arising system
(Figure~\ref{figure:particle-pairs:cluster-cost}).
The coincidence of expensive contact resolution with phases of low
concurrency explains the limited weak scalability.

\paragraph{The tower setup.}

\begin{figure}[htb]
 \begin{center}
  \includegraphics[width=0.24\textwidth]{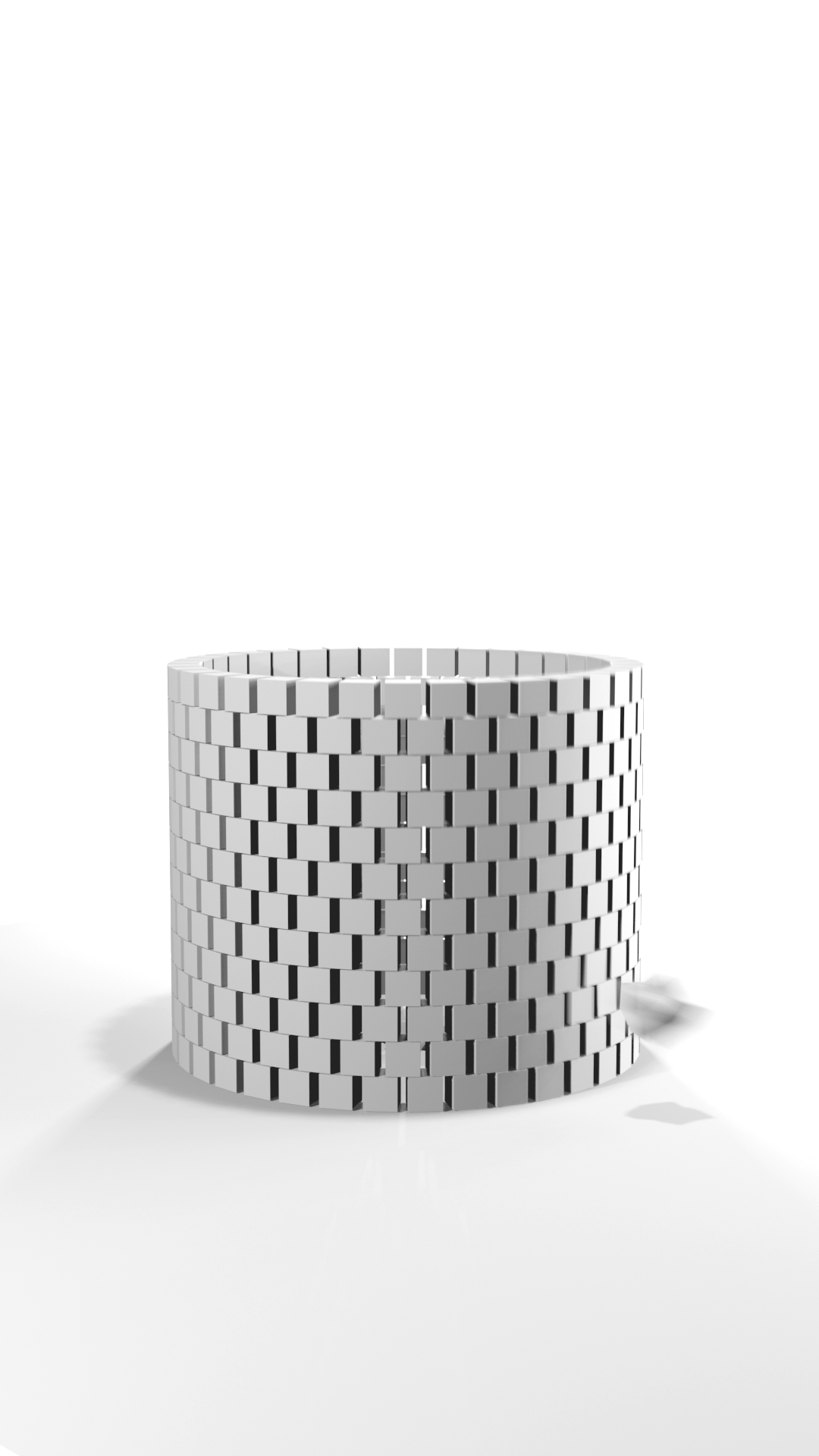}
  \includegraphics[width=0.24\textwidth]{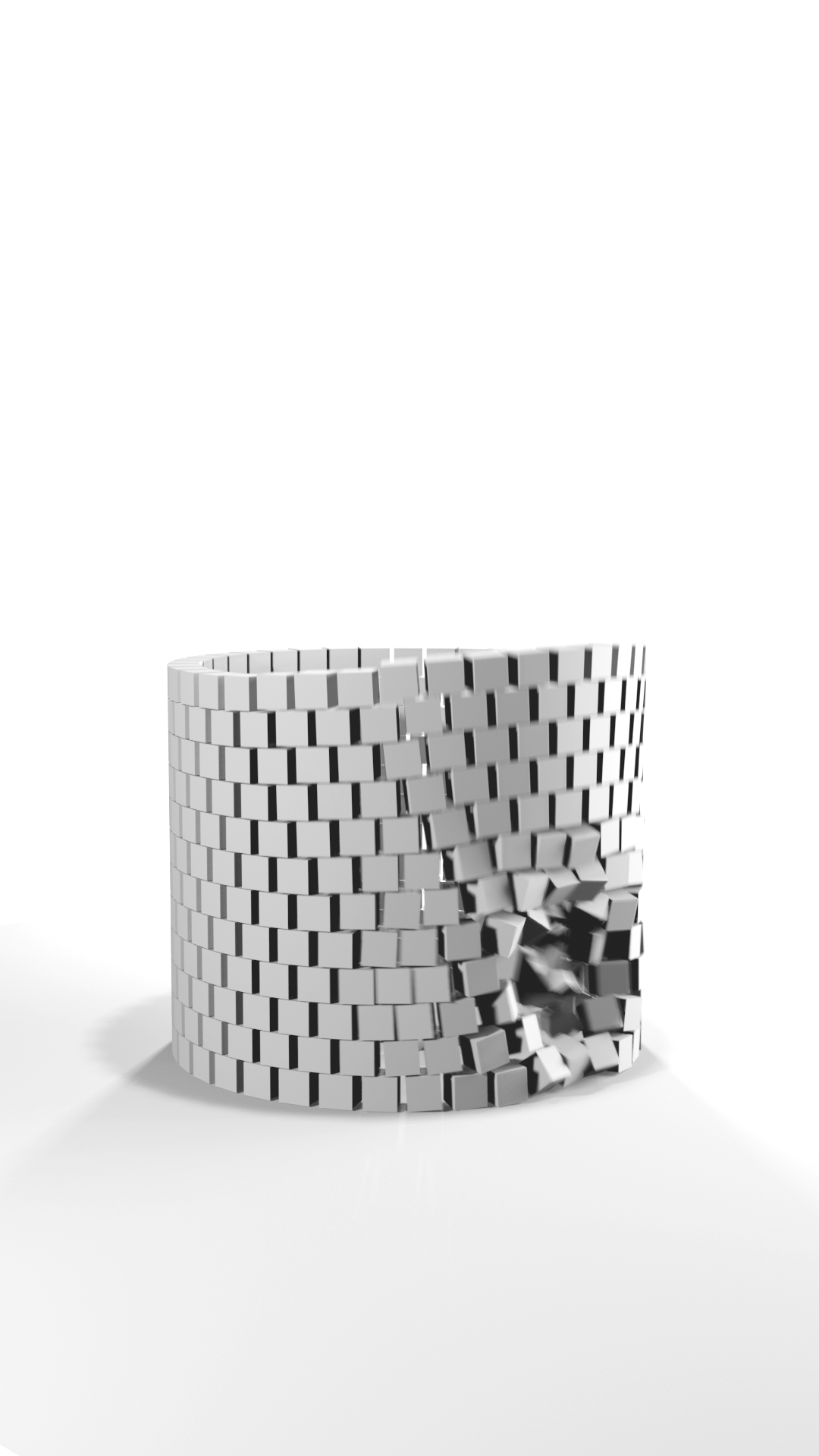}
  \includegraphics[width=0.24\textwidth]{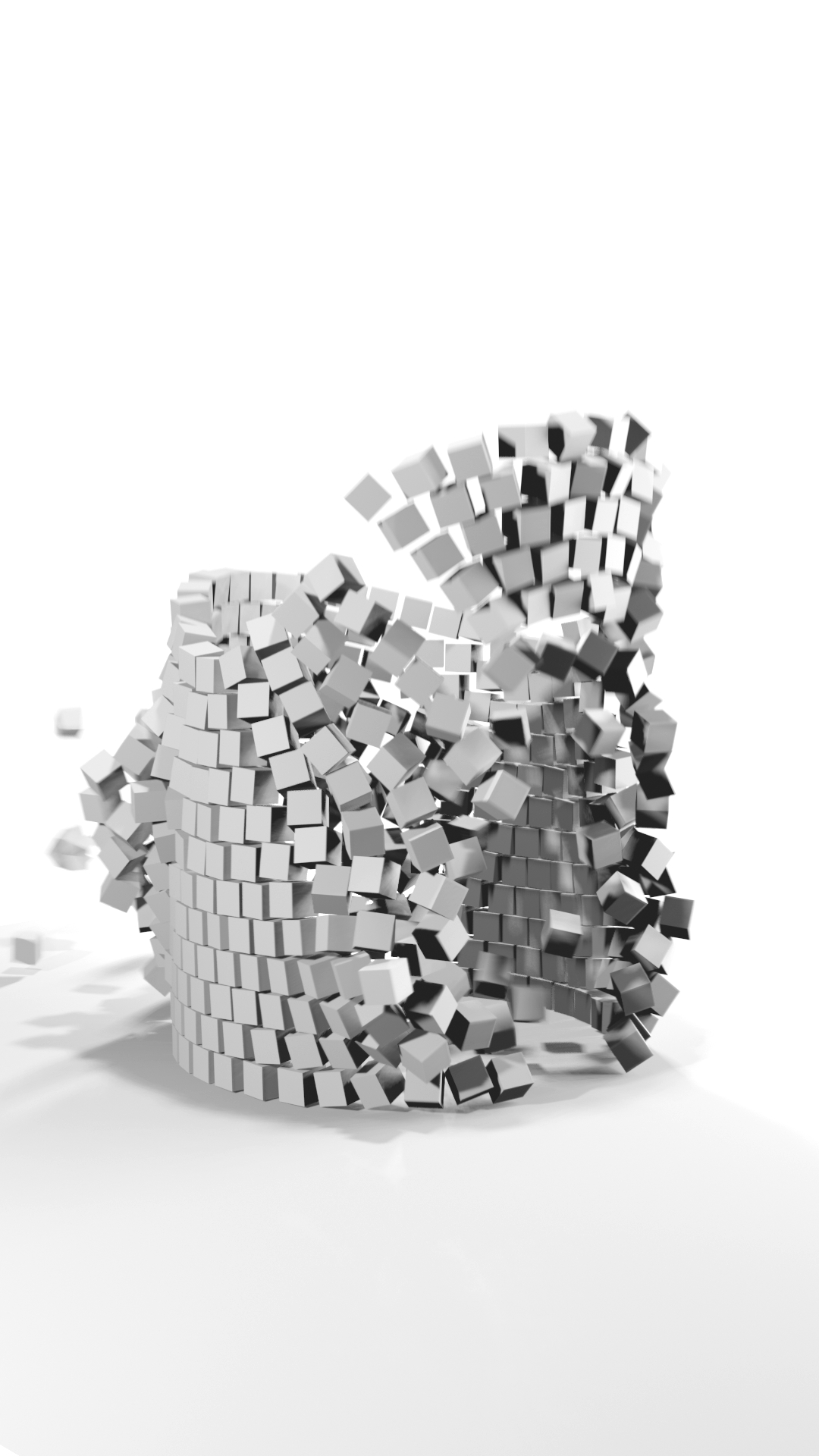}
  \includegraphics[width=0.24\textwidth]{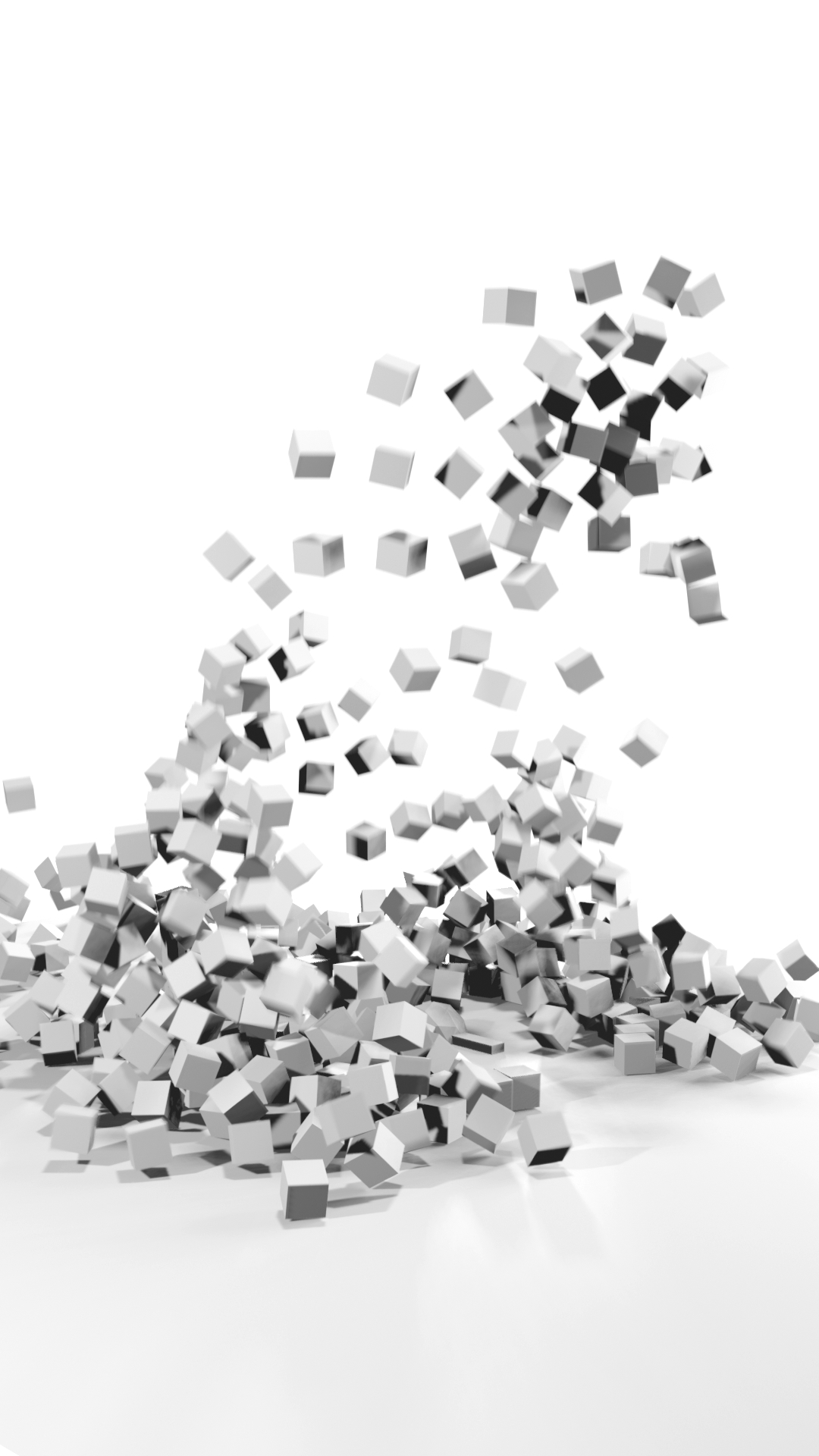}
 \end{center}
 \caption{
  The tower scenario from left to right:
  Rings of cubes form a steady tower, when a single sphere-like particle
  approaches its base with a high speed.
  The fast moving particle destroys the tower.
  \label{figure:results:particle-stacks:tower-setup}
 }
\end{figure}

We confirm our observations through the particle-tower setup, where we
arrange rings of 32 cubes into layers on top of each other. 
Similar to bricks, they yield a tower-like construction.
A small, heavy and very fast moving object crashes into the bottom of the tower
and makes it collapse (Figure~\ref{figure:results:particle-stacks:tower-setup}).

Initially, the tower is in a steady state.
Its bricks (particles) do not move.
As the projectile hits the tower, a shock travels through the structure as the
particles are incompressible.
Cubes in the tower use $|\triangleSet|=12$, the ground plane uses $|\triangleSet|=2$ and the
projectile uses $|\triangleSet|=80$.
We scale this experiment by increasing the height of the towers.

\begin{observation}
 The tower yields a computational character that is an extreme case of the
 stacks: Two clusters ``suddenly'' interact and the simple two-cluster topology
 is completely destroyed, before the simulation settles.
\end{observation}

\noindent
With a dense packing of the tower bricks, the setup is either (almost) at rest
or the shock has separated off particles already.
Compared to the particle stacks, we have a simpler topology, and this
topology remains intact over a longer time span.
Therefore, the CPU usage of the local time stepping follows the global one, and
we get close to no scalability
(Figure~\ref{figure:results:particle-stacks:tower:global_vs_local_towers}).
As we only enrich one cluster with more and more particles, we see a linear
increase of the compute cost for both local and global time stepping. 
For reasonably large setups, we gain performance once we add a second
memory controller; which is stereotypical for bandwidth-bound codes.

\begin{figure}[htb]
 \begin{center}
  \includegraphics[width=0.45\textwidth]{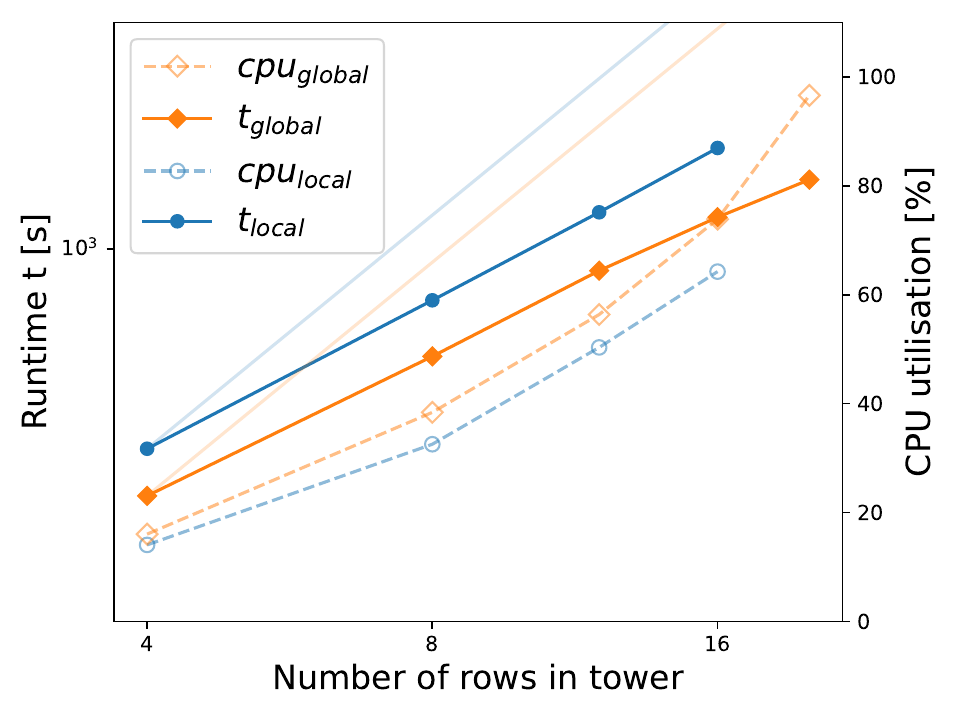}
  \includegraphics[width=0.45\textwidth]{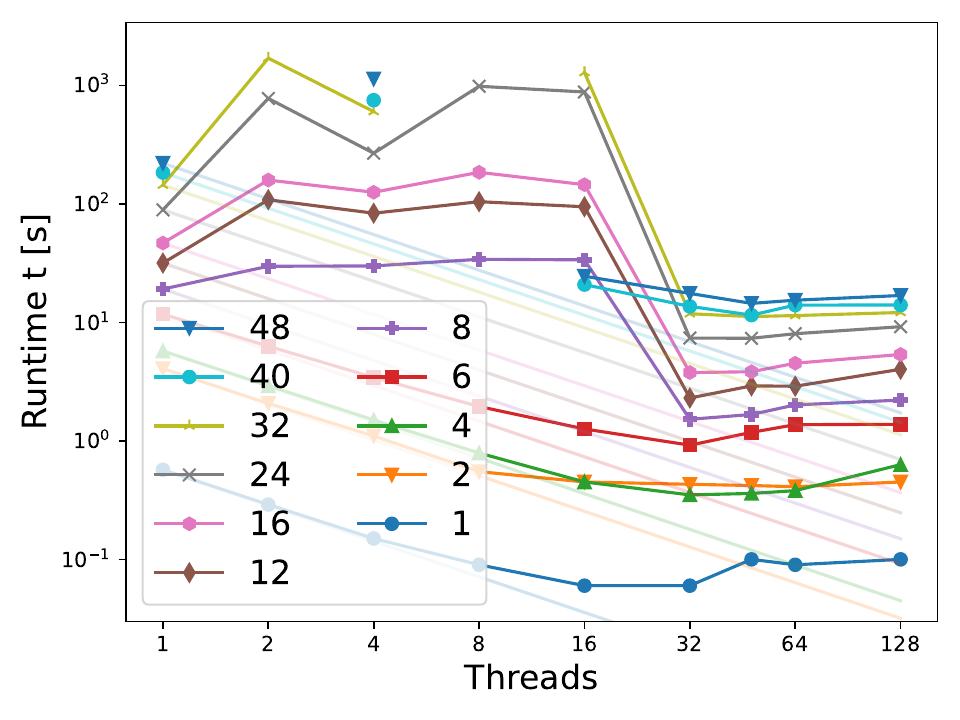}
 \end{center}
 \caption{
 Left: Benchmarking of local time stepping vs.~global time stepping on a whole
 node.
 Right: Various strong scaling measurements for the tower setup.
 Additional particles result from additional layers (rings) on top of the tower.
  \label{figure:results:particle-stacks:tower:global_vs_local_towers}
 }
\end{figure}

Just after the single projectile particle has hit the tower, we obtain very
small time steps and the tower's cubes hence decompose into various clusters.
At this point, clusters start to advance with slightly differing time steps.
However, the geometric topology changes again, particles are rolled back and we
hence do, compute time wisely, not benefit from the arising concurrency.

\paragraph{Discussion}
Our algorithmic design is guided by the geometric identification of clusters.
If such a cluster topology ceases to exist, we can not exploit a spatial
decomposition anymore to obtain a scaling code.
We also loose the opportunity to march different particles in time at different
speeds.
It is reasonable to assume that composites of many particles require a different
algorithmic approach compared to the techniques presented in the present paper.
Notably, it makes sense in such a case to work with larger, globalised clusters
and hence to avoid any overhead imposed through local time stepping.

\subsection{The hopper and the staircase}

\begin{figure}[htb]
 \begin{center}
  \includegraphics[width=0.24\textwidth]{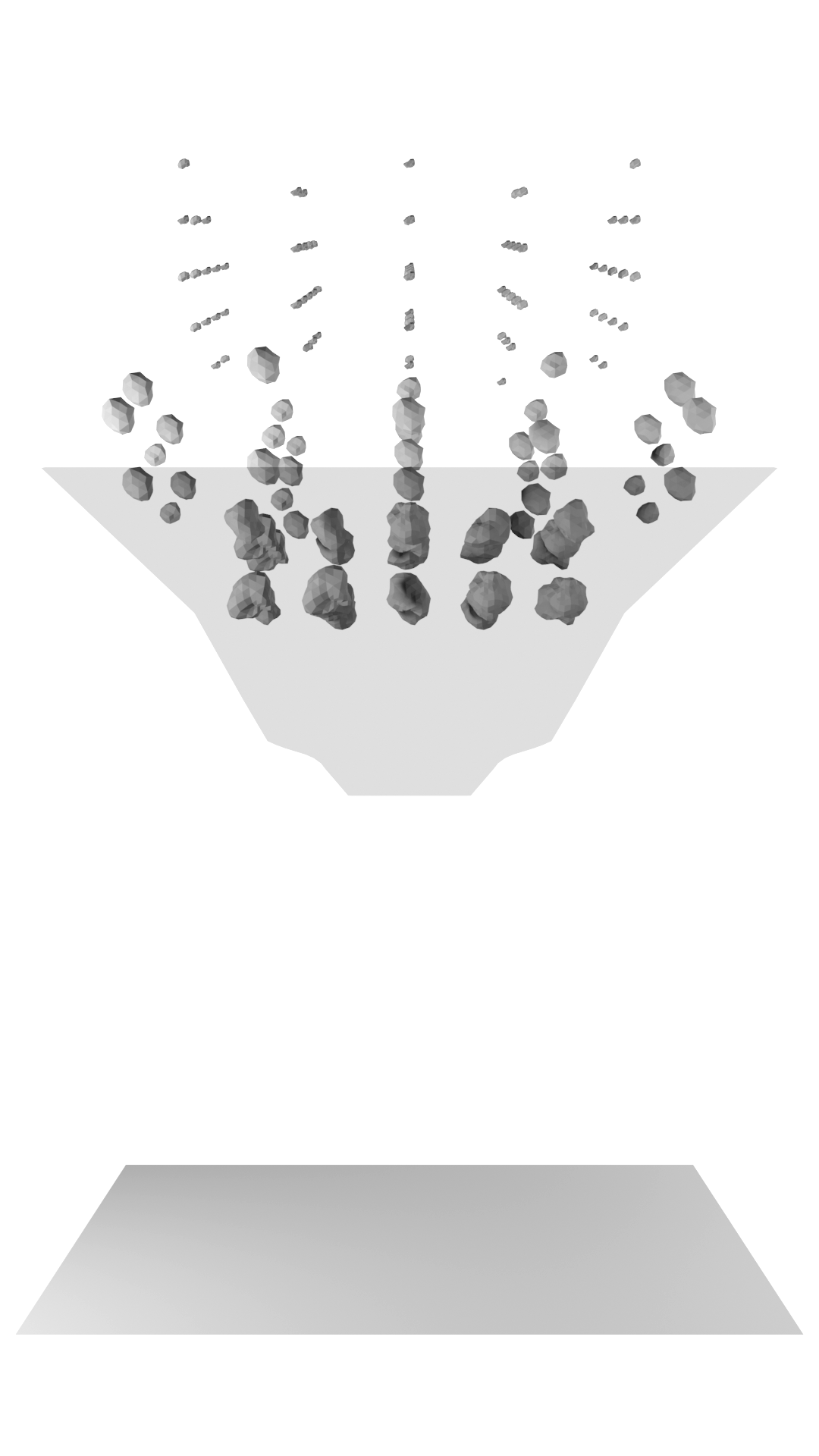}
  \includegraphics[width=0.24\textwidth]{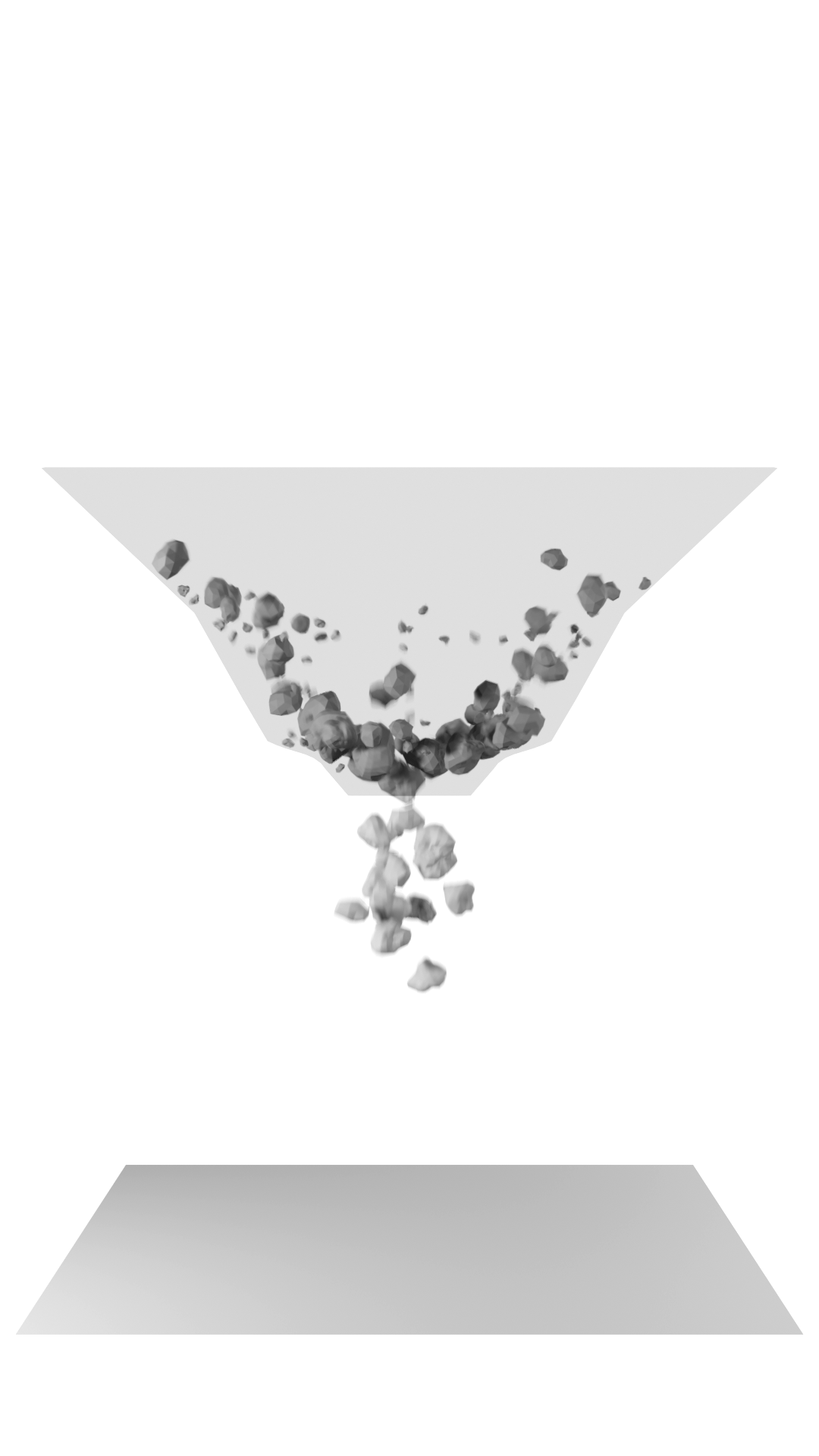}
  \includegraphics[width=0.24\textwidth]{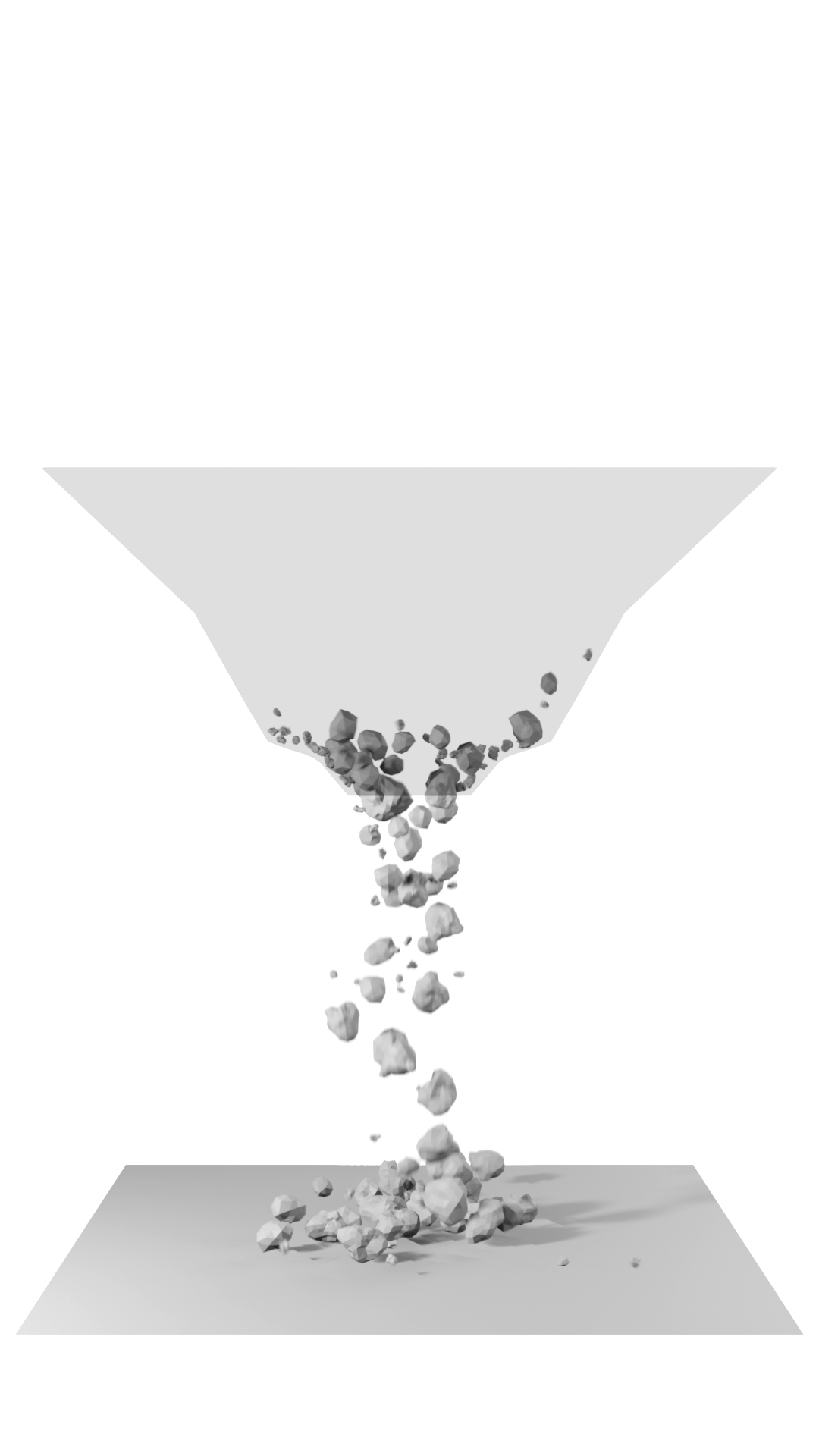}
  \includegraphics[width=0.24\textwidth]{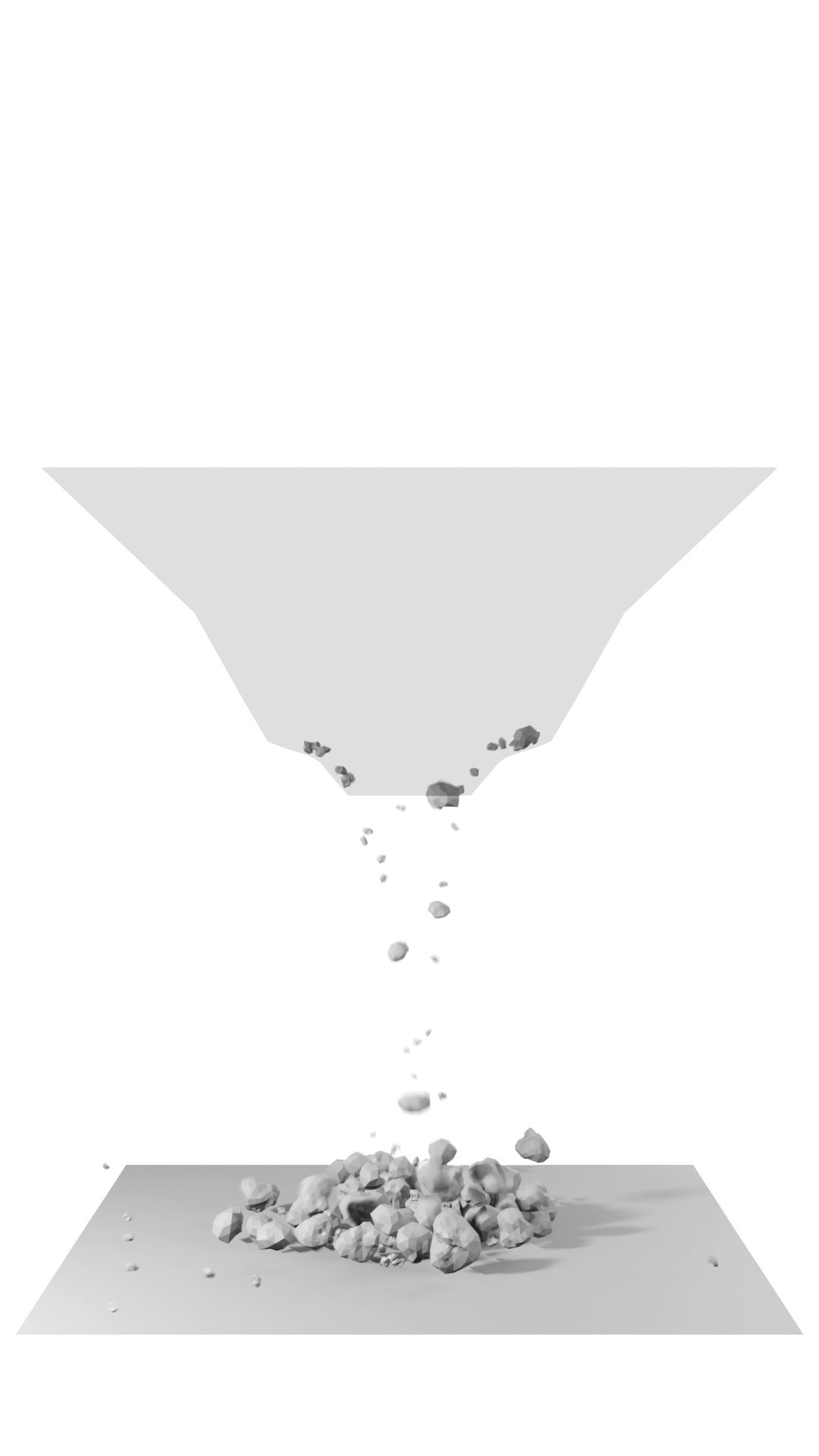}
 \end{center}
 \caption{
  Hopper setup from left to right:
  The initial state with particles are released above a hopper.
  They drop into the hopper where they block each other.
  Eventually, the hopper empties to create a pile directly underneath.
  \label{figure:results:hopper:setup}
 }
\end{figure}

We close our experimental studies with two sophisticated setups:
In our hopper scenario,
a collection of particles is dropped
through a hopper and comes to rest in a pile on a plane underneath.
Dropped particles begin in free fall with no interaction with surrounding
particles before they enter the hopper.
As they bounce from the hopper walls, and as more and more particles fall into
the mouth of the hopper over time, they interact both with each other and the
hopper's walls (Figure~\ref{figure:results:hopper:setup}).
Our hopper model consists of $|\particleSet|=1,856$ triangles and its opening at
the bottom can accommodate around four of the largest particles at the same
time.

\begin{figure}[htb]
 \begin{center}
  \includegraphics[width=0.24\textwidth]{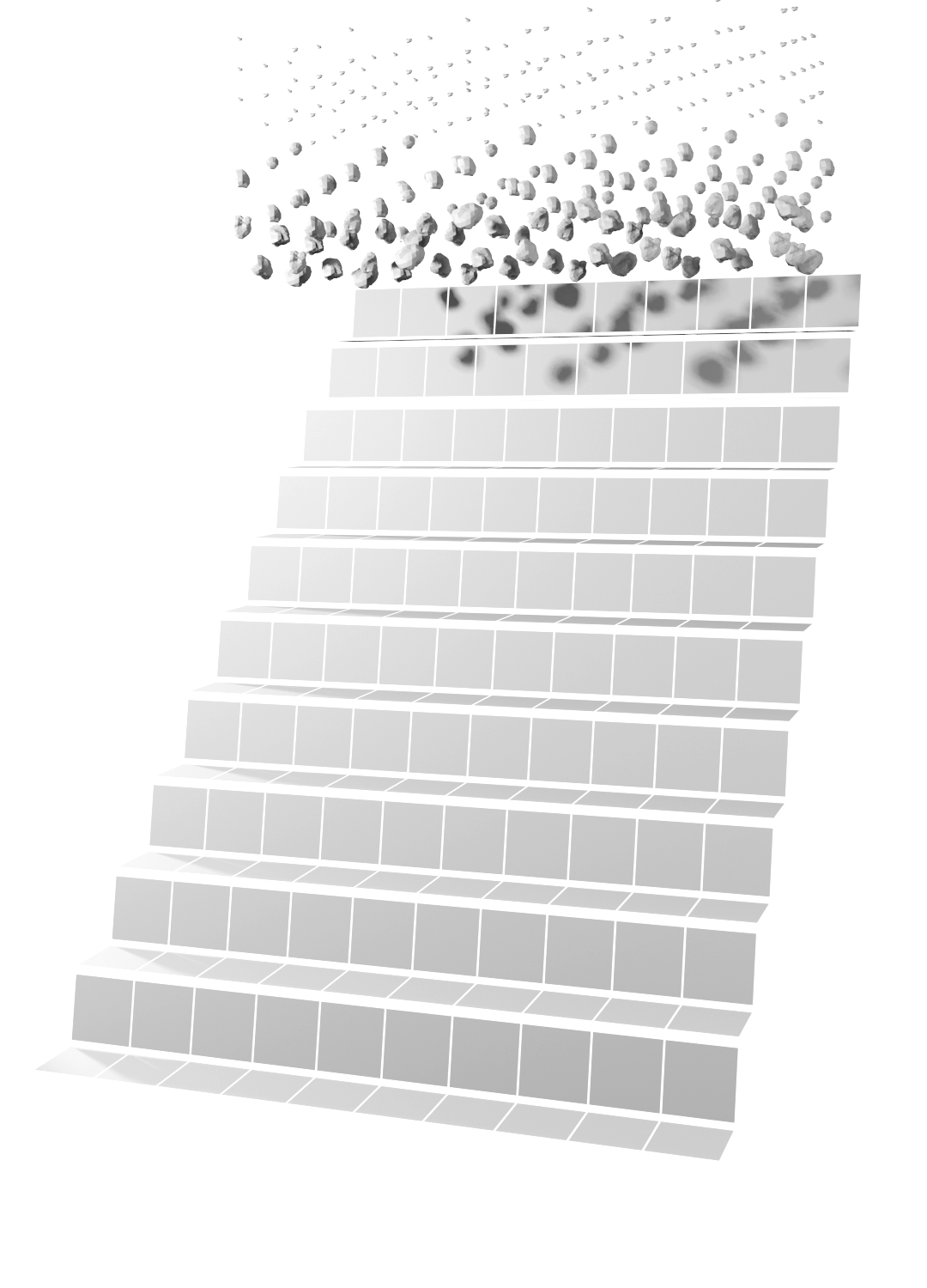}
  \includegraphics[width=0.24\textwidth]{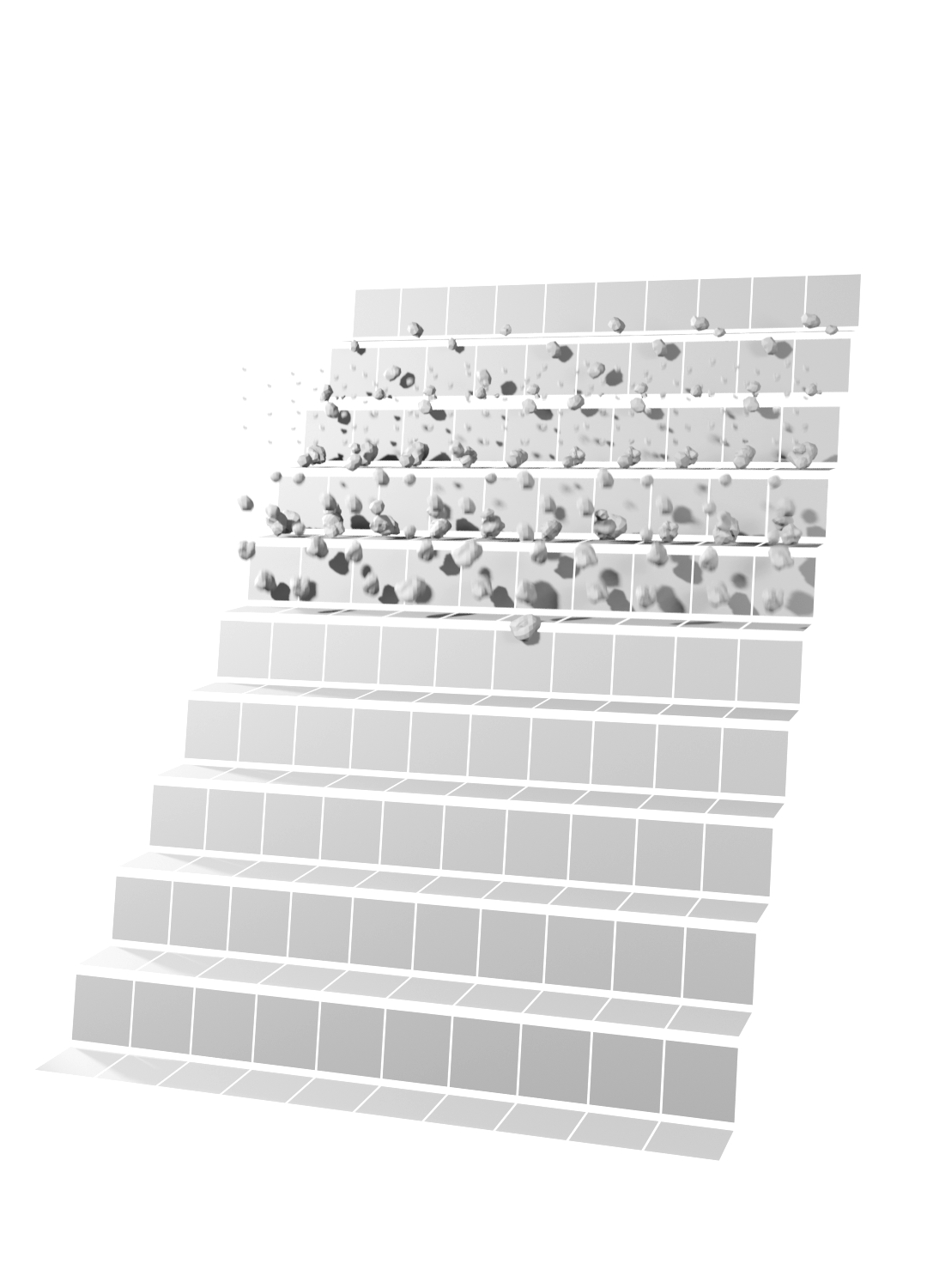}
  \includegraphics[width=0.24\textwidth]{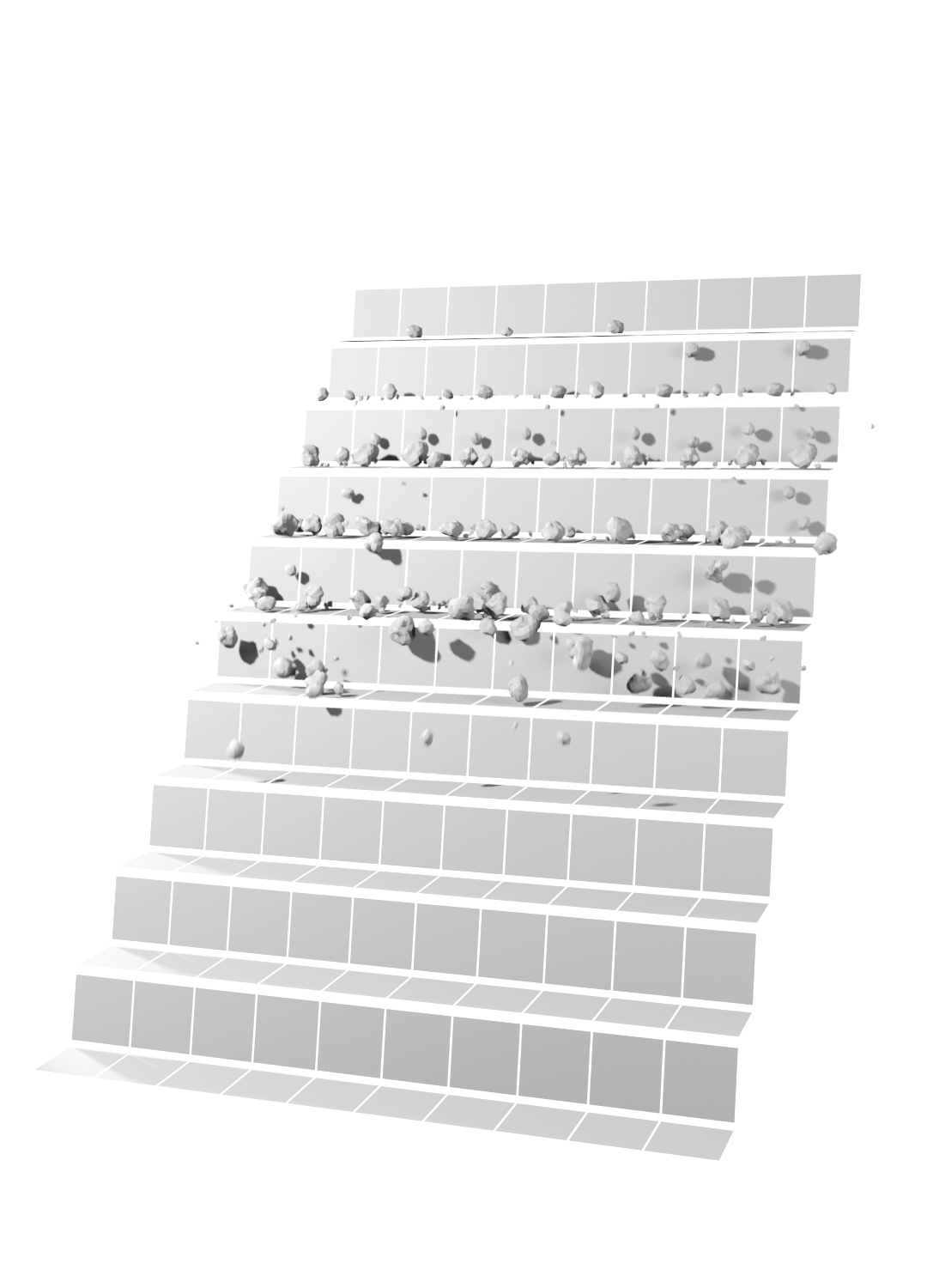}
  \includegraphics[width=0.24\textwidth]{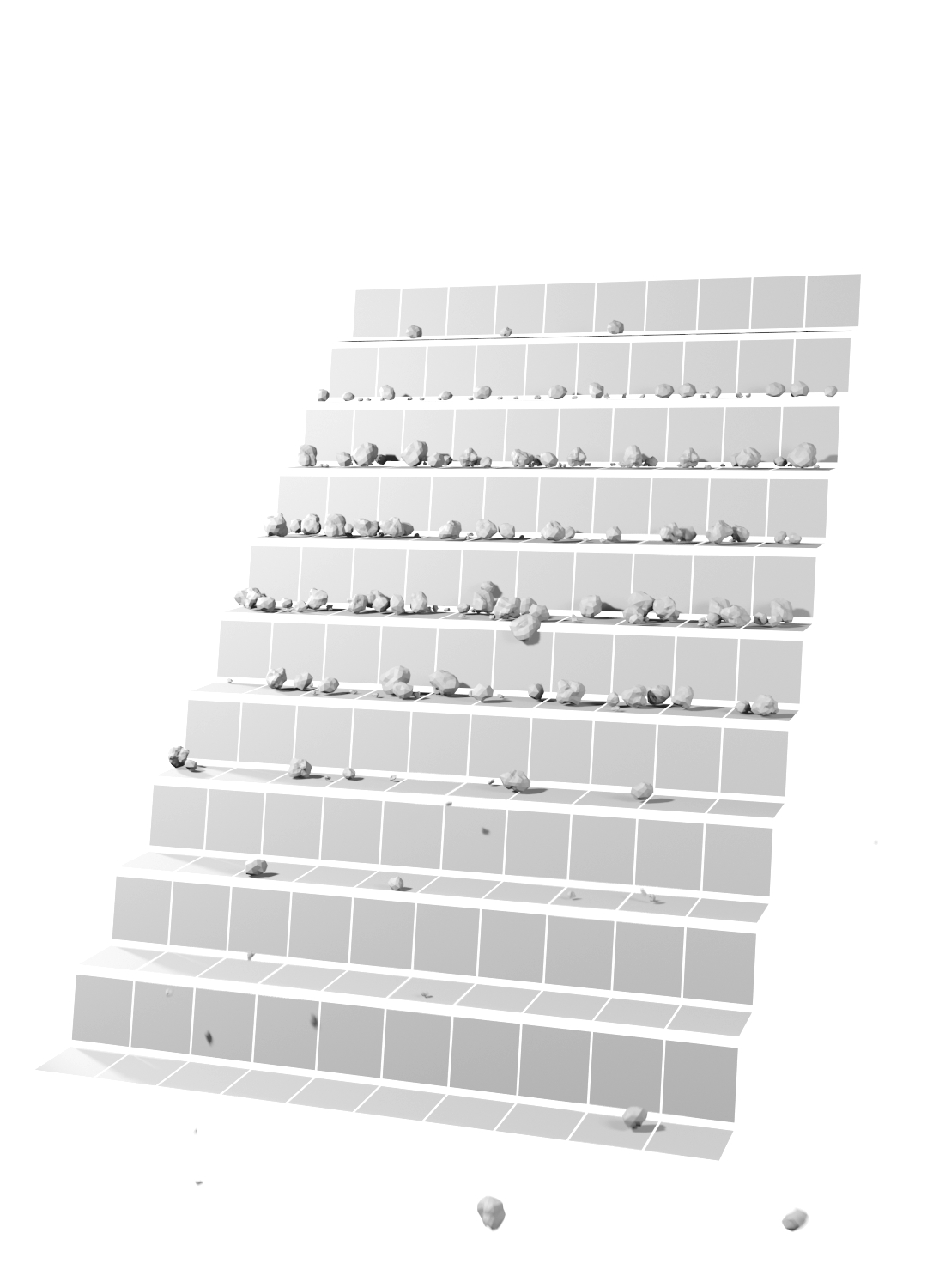}
 \end{center}
 \caption{
  Particles tumble down a staircase.
  \label{figure:results:hopper:staircase-setup}
 }
\end{figure}

In our staircase setup, we let the particles hop down a staircase (Figure~\ref{figure:results:hopper:staircase-setup}).
The staircase is made up of 20 planes, each consisting of two triangles.
Whenever particles hit the staircase, they either come to rest at their step or
continue toppling downwards.
They can hit further particles and knock them off their respective
steps where they might have come to rest already.
Eventually, all particles are either settled on the flat steps
or have fallen over the side or bottom.

In both scenarios, we use sphere-like, triangulated particle shapes.
They result from a triangulation of a sphere subject to
random, hierarchical noise:
We decompose the sphere equidistantly into
$|\mathbb{T}|$ triangles.
The distance of the triangles' vertices on the sphere is 
added Perlin noise, which offsets the vertex along the normal direction
of the surface.

If $r$ is the radius of the original sphere, we end up with distances from 
$
[r,r \cdot \eta _r]
$, i.e.~$\eta _r=1$ adds no noise and thus yields a perfect, triangulated sphere.
The higher $\eta _r $ the more ``degenerated'' the particle featuring
both low and high frequency geometry distortions on the surface.
Manipulating $r_{\text{min}} \leq r \leq r_{\text{max}}$ and $\eta _r$, we 
randomise the size and shape distribution of the particles \cite{multires2022}.

The quasi-spherical particles range from small grains ($|\particleSet|=80$ and
a radius of $r=0.025$) to larger rocks ($|\particleSet|=320$ with a radius of
$r=1$).
We run the simulation only for the time until the last particle has hit the 
hopper or all particles come to rest on the staircase or have fallen off the
last step.

\begin{figure}[htb]
 \begin{center}
  \includegraphics[width=0.45\textwidth]{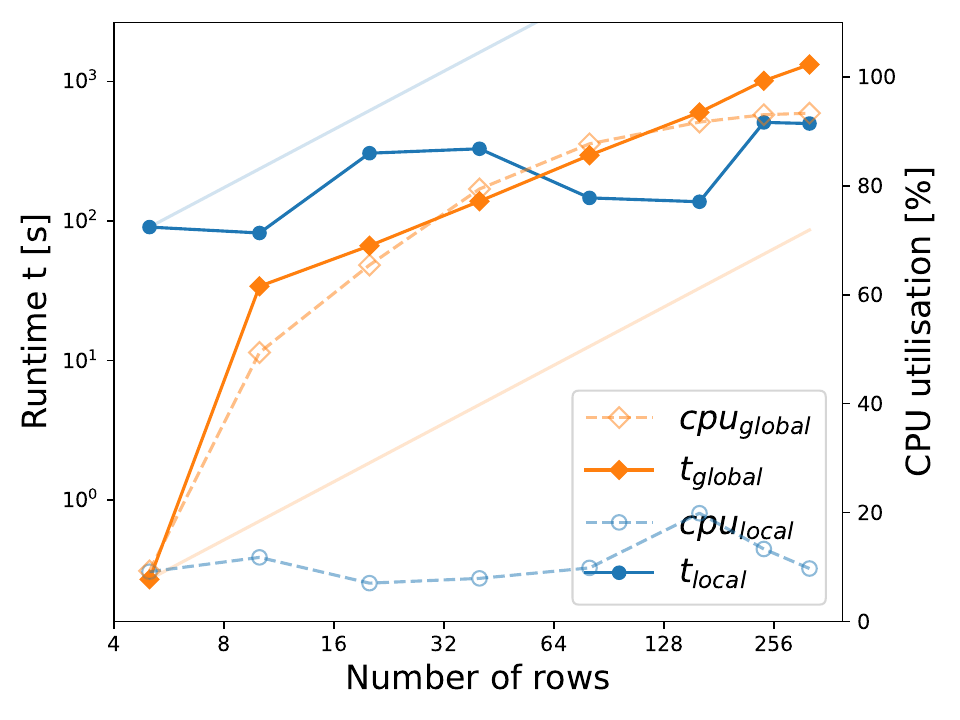}
  \includegraphics[width=0.45\textwidth]{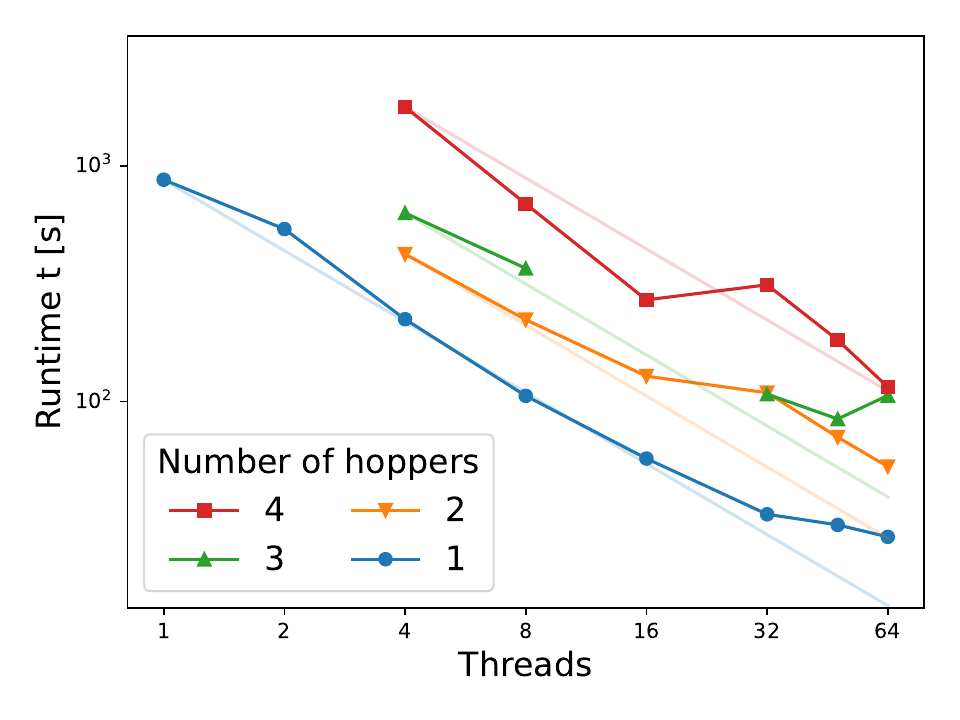}
 \end{center}
 \caption{
 Runtime and CPU utilisation for the staircase setup on the whole node (left)
 and strong scaling curves for multiple copies of the hopper setup (right).
 The right figure presents data for various numbers of particle, whereas
 the left one studies only the smallest particle count.
  \label{figure:results:staircase:runtime}
 }
\end{figure}

\begin{observation}
 For both setups, the local time stepping struggles to keep pace with global
 time stepping unless the setups grow large. 
 Strong and weak scalability are good. 
\end{observation}

\paragraph{Cluster topology, active clusters and concurrency level.}

Both simulations kick off with around $|\particleSet |$ clusters.
As
the particles assemble within the hopper or pile up on the staircase, we obtain
larger clusters.
This is counteracted by small step sizes from complex interactions reducing the
possible interaction range of any single particle in a given step.
A global time stepping scheme hence suffers from the stiffness of individual
particle interactions, whereas the local time stepping always manages to advance
some clusters aggressively forward in time. 
Nevertheless, the permanent topology changes and rollbacks imply that the local
time stepping's overhead is very big.
We have to simulate a significant number of particles to make the local time
stepping overtake its global cousin
(Figure~\ref{figure:results:staircase:runtime}).

The free fall ahead of any particle hitting the geometry, or a free fall once
the particles leave the hopper resembles the behaviour in previous studies.
These phases scale.
The computationally expensive middle phase does not scale that well, and our 
scaling curves become non-smooth, reflecting the anarchic rearrangement of
clusters.
The algorithm's concurrency changes permanently.

\paragraph{Time step size distribution.}

The free fall prior to any particle hitting the obstacle allows the 
particles to advance at the same pace independent of local or global time
stepping.
As all particles start with a hard-coded initial downfall velocity (we omit
gravity), we can predict a reasonable initial $\timeStepSize$.

\begin{figure}[htb]
 \begin{center}
  \includegraphics[width=0.45\textwidth]{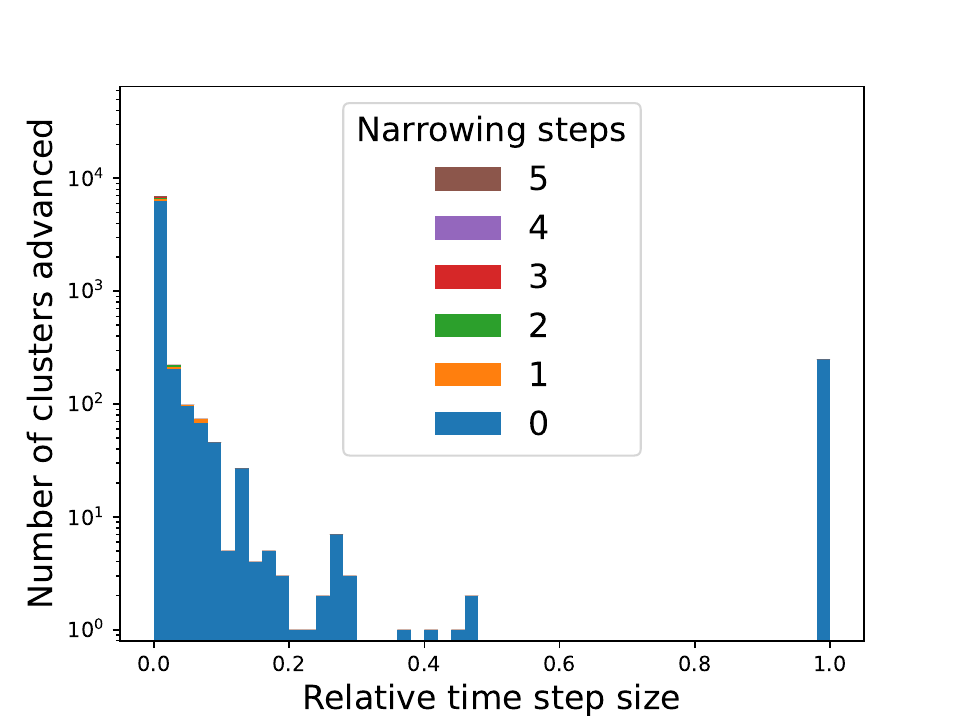}
  \includegraphics[width=0.45\textwidth]{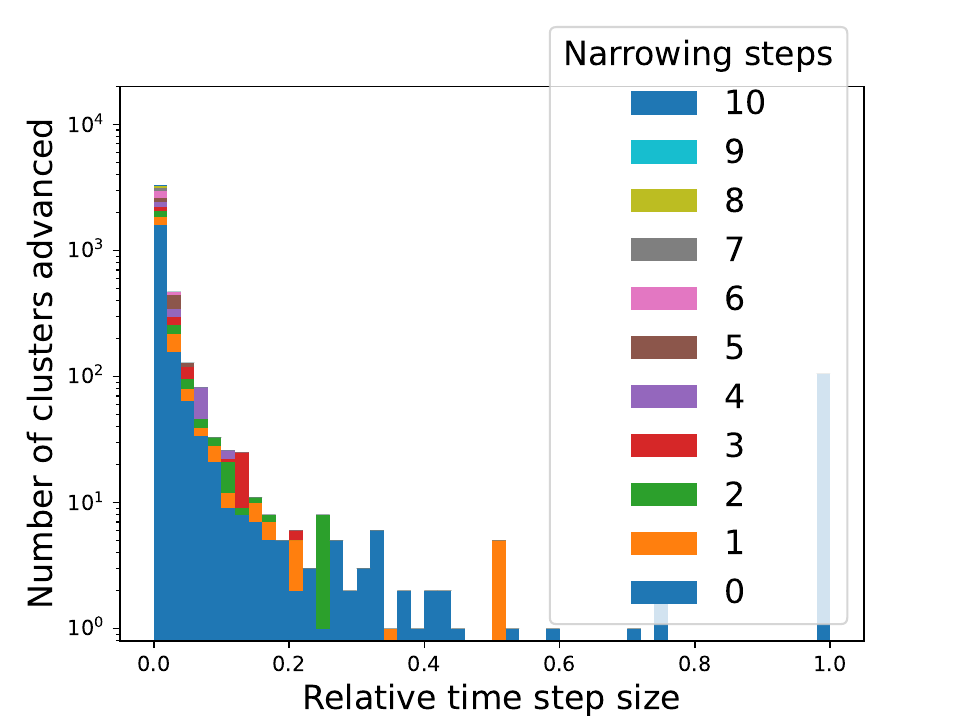}
 \end{center}
 \caption{
   Time step size distribution for the hopper (left) and the staircase (right)
   setup using local time stepping.
  \label{figure:particle-hopper:time_step_size}
 }
\end{figure}

Once the particles hit the hopper or staircase, we get a more complicated time
step size distribution (Figure~\ref{figure:particle-hopper:time_step_size}):
Clusters advance with different time step sizes, while notably the
staircase enforces quite a lot of narrowing. 
We obtain huge clusters at time, where we have to try out smaller and
smaller time step sizes, before the algorithm finds an effective
time step size.

\begin{figure}[htb]
 \begin{center}
  \includegraphics[width=0.45\textwidth]{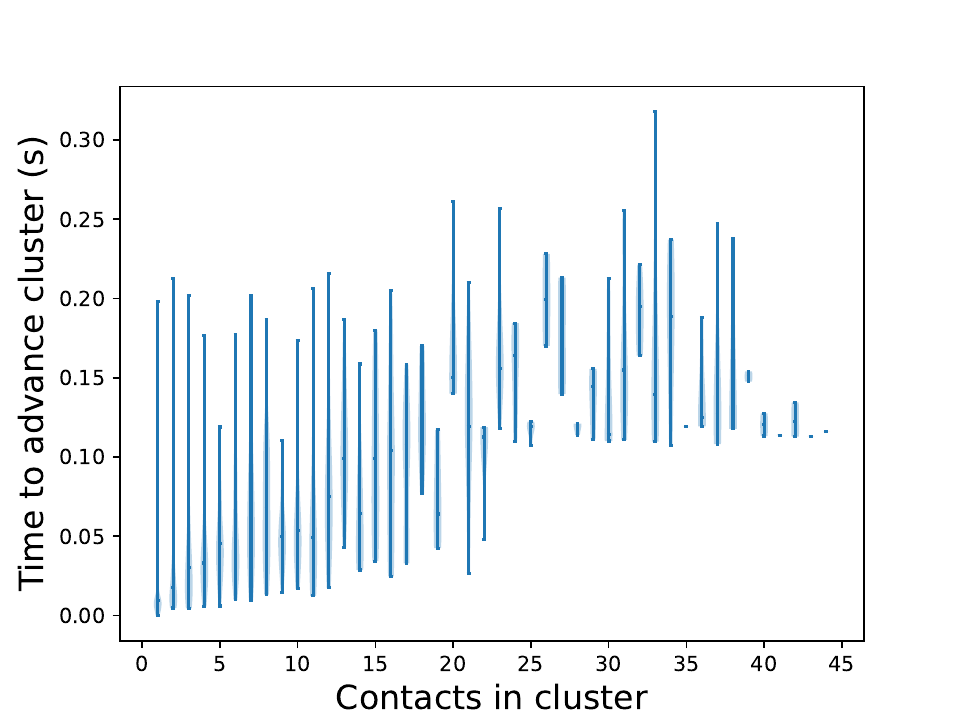}
  \includegraphics[width=0.45\textwidth]{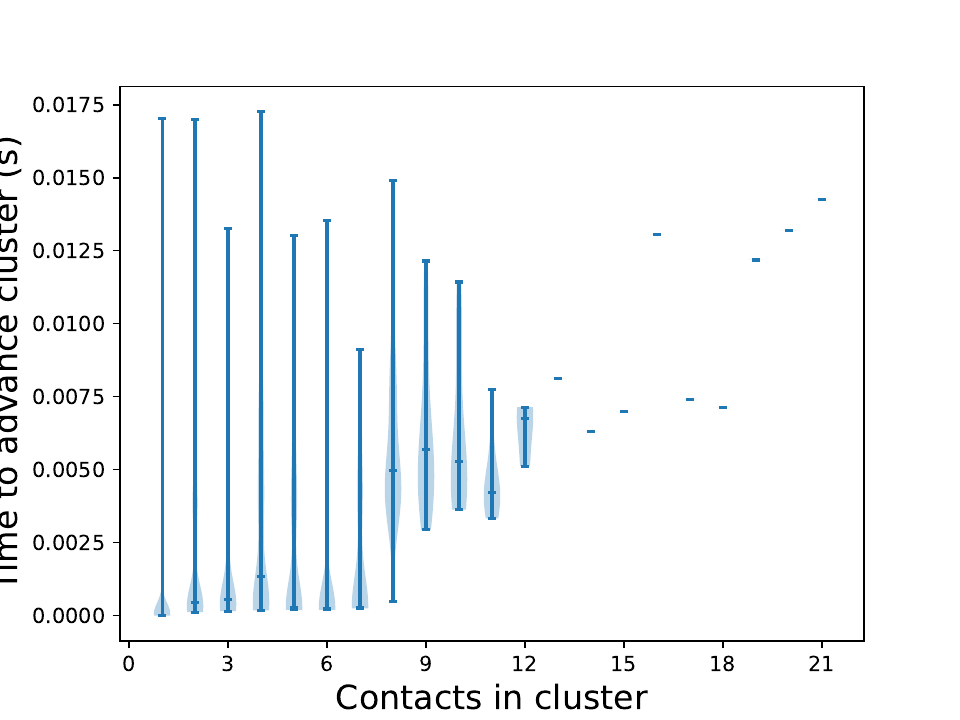}
 \end{center}
 \caption{
   The local time stepping's cost per contact resolution for the hopper (left) and the staircase (right).
  \label{figure:particle-hopper:cost_per_contact}
 }
\end{figure}

\paragraph{Contact points and particle interactions.}

With complex particle geometries, the cost to solve the momentum exchange
equations becomes unpredictable
(Figure~\ref{figure:particle-hopper:cost_per_contact}).
The staircase, on average, produces fewer contacts per cluster, as the clusters
are smaller.
Yet, it becomes impossible to see a clear correlation between the number of
contact points per cluster and the time required to solve the arising equation
system.
Due to a comparison with previous results, we can conclude that the complex
geometric shapes have to cause this effect.

%
%

\section{Conclusion}
\label{section:conclusion}

%
%
Local time stepping can be characterised in several ways:
First, we have to decide between optimistic and
conservative time stepping.
Optimistic time stepping follows a trial-and-error approach.
We start with a time step that is as large as possible, but, in the case of
failure, i.e.~penetration of physically incompressible objects, roll back,
reduce the time step size and recompute until we have found an admissible step size.
In conservative time stepping, we underestimate the admissible time
step size all the time such that we can guarantees an
admissible step size without any penetration.
Second, local time stepping can be characterised through its
granularity.
Individual particles could progress with their own time step size,
or we could cluster them and make the clusters advance in time with the
same time step size.
Global time stepping equals a degenerated cluster approach, where one cluster
comprises all rigid bodies within the system.
Finally, we have to differentiate local time stepping which allows arbitrary
time step sizes from bucketing approaches, where sizes are picked from a
finite set of discretised time slaps \cite{localSPH,localTimeStepDEM}.
The latter approach naturally leads into subcycling.

Our approach strives for a compromise between the extremes,
constructing a bespoke, efficient variant:
The clustering in combination with the narrowing yields a pessimistic approach,
while the global scheme remains optimistic and hence requires rollbacks
occasionally.
We try to advance clusters hosting many particles in one rush, yet allow each
particle to join another cluster in every single time step.
At the same time, our flavour of local time stepping does not discretise time
into fixed subintervals.
As we consolidate particles of one cluster such that they start from the same
time stamp and advance with the same time step size, we however constrain the
progress in time such that particles do not advance totally anarchically.

Throughout the involved algorithmic subsets, we employ a multiscale language to
speed the calculations up:
Particles are represented by different shapes of increasing geometric
complexity, we look at problems within a space-time setting, where we zoom into
the time region of interest, and we employ various geometric representations to
rule out potential overlaps of the particles' $\halo$-regions early.
This way, we can harvest the efficiency promised by local time
stepping---particles far away from potential collisions advance quickly \cite{adaptiveLocal18}.

While the efficiency gains obtained due to local time stepping are impressive,
one major challenge is omitted:
Load balancing local time stepping schemes is difficult.
Without subcycling and any fixed cluster association, our algorithm
yields algorithmic phases of varying concurrency all the time, and the load per
compute step is difficult to predict.
On a shared memory node, we rely on work stealing to level out imbalances. 
This works yet requires sufficient workload per step, which we do not obtain in
all scenarios.

Two algorithmic degrees of freedom to address this issues are not
further discussed in the present paper:
DEM would allow us to mask out clusters artificially.
In phases with many active clusters, it might be reasonable not to advance all
algorithms, but to rule out some clusters which are then updated later when
otherwise only few clusters are active.
At the same time, it might be reasonable to abandon the synchronisation of the
compute phases.
We run through the six compute steps per time step one by one.
Future implementation might want to try to overlap the individual steps.

Future work also will have to develop algorithms that scale beyond compute node
boundaries.
This could either be achieved through tasking approaches which allow for task
migration beyond nodes.
Alternatively, it might be reasonable to use larger clusters spanning multiple
nodes.
As highlighted, global time stepping is equivalent to working with one cluster
only in our formalism.
Excessively large clusters introduce a hybrid between local and global time
stepping, and might be key to scale up local time stepping schemes over parallel
machines.

\ifthenelse{ \boolean{useSISC} }{
}{
  \section*{Acknowledgements}

The work was funded by an EPSRC DTA PhD scholarship (award no. 1764342).
It made use of the facilities of the Hamilton HPC Service of Durham University.
The research aligns and has been supported by EPSRC's Excalibur
programme through its cross-cutting project EX20-9 \textit{Exposing Parallelism: Task Parallelism}
(Grant ESA 10 CDEL) and the DDWG project \textit{PAX--HPC} (Gant EP/W026775/1).
Tobias' group appreciates the support by Intel's Academic Centre of
Excellence at Durham University.

}

%
%

\ifthenelse{ \boolean{useSISC} }{
  \bibliographystyle{siamplain}
}{
  \bibliographystyle{plain}
}

\bibliography{paper}

\end{document}